\documentclass[12pt,a4paper]{article}
\usepackage{amsmath}
\usepackage{amssymb}
\usepackage{amsthm}
\usepackage{amsfonts}
\usepackage[all]{xy}
\usepackage{color}
\usepackage{tikz}

\begin{document}

\title{Broadcasts on Paths and Cycles}

\author{Sabrina BOUCHOUIKA~\thanks{Faculty of Mathematics, Laboratory L'IFORCE, University of Sciences and Technology
Houari Boumediene (USTHB), B.P.~32 El-Alia, Bab-Ezzouar, 16111 Algiers, Algeria.}
\and Isma BOUCHEMAKH~\footnotemark[1]
\and \'Eric SOPENA~\thanks{Univ. Bordeaux, CNRS, Bordeaux INP,  LaBRI, UMR 5800, F-33400, Talence, France.}
}
\maketitle

\begin{abstract}
A broadcast on a graph $G=(V,E)$ is a function  $f: V\longrightarrow \{0,\ldots,\operatorname{diam}(G)\}$ such that 
$f(v)\leq e_G(v)$ for every vertex $v\in V$, where
$\operatorname{diam}(G)$ denotes the diameter of $G$ and $e_G(v)$ the eccentricity of $v$ in $G$. 
The cost of such a broadcast is then the value $\sum_{v\in V}f(v)$.
Various types of broadcast functions on graphs have been considered in the literature, 
in relation with domination, irredundence, independence
or packing, leading to the introduction of several broadcast numbers on graphs.

In this paper, we determine these broadcast numbers for all paths and cycles, thus answering a question
raised in [D.~Ahmadi, G.H.~Fricke, C.~Schroeder, S.T.~Hedetniemi and R.C.~Laskar, 
Broadcast irredundance in graphs. {\it Congr. Numer.} 224 (2015), 17--31].
\end{abstract}

\newtheorem{theorem}{Theorem}[section]
\newtheorem{lemma}[theorem]{Lemma}
\newtheorem{conjecture}[theorem]{Conjecture}
\newtheorem{observation}[theorem]{Observation}
\newtheorem{claim}{Claim}
\renewcommand{\theclaim}{\Alph{claim}}
\newtheorem{corollary}[theorem]{Corollary}
\newtheorem{proposition}[theorem]{Proposition}
\newtheorem{question}[theorem]{Question}

\medskip


\medskip


\newcommand\EFFACE[1]{}

\def\diam{{\rm diam}}
\def\rad{{\rm rad}}
\def\cartesian{{\,\square\,}}

\section {Introduction}

Let $G=(V,E)$ be a graph of {\it  order} $n=|V|$ and {\it  size} $m=|E|$. 
The {\it open neighborhood} of a vertex $v \in V$ is the set $N_G(v) = \{u: uv \in E \}$ of vertices adjacent to $v$. 
Each vertex $u \in N_G(v)$ is a {\it neighbor} of $v$ in $G$.  
The {\it  closed neighborhood} of $v$ is the set $N_G[v] = N_G(v) \cup \{v\}$. 
The {\it open neighborhood} of a set $S \subseteq V$ of vertices is $N_G(S) = \cup_{v\in S} N_G(v)$, while the {\it closed neighborhood} of $S$ is the set $N_G[S]=N_G(S)\cup S$.
The {\it degree} of a vertex $v$ in $G$,  denoted  $\deg_G(v)$, is the size of the open neighborhood of $v$. 

A {\it $(u,v)$-geodesic} in a graph $G$ is a shortest path joining $u$ and $v$.
We denote by $d_G(u,v)$ the {\it distance} between the vertices $u$ and $v$ in $G$, that is, the length of a $(u,v)$-geodesic in $G$.
The {\it  eccentricity} $e_G(v)$ of a vertex $v$ in $G$ is the maximum distance from $v$ to any other vertex of $G$.
The {\it radius} $\rad(G)$ and the {\it diameter} $\diam(G)$ of a graph $G$ are  the minimum and the maximum eccentricity among the vertices of $G$,  respectively.

\medskip

A function $f: V \longrightarrow \{0,\ldots,\operatorname{diam}(G)\}$ is a {\it  broadcast} on a graph $G=(V,E)$ if $f(v)\leq e_G(v)$ for every vertex $v\in V$. 
The value $f(v)$ is called the {\it $f$-value} of $v$.
An {\it  $f$-broadcast vertex} (or an {\it  $f$-dominating vertex}) is a vertex $v$ for which $f(v)>0$. 
The set of all $f$-broadcast vertices  is  denoted $V^+_f(G)$. 
If $v\in V^+_f$ is an $f$-broadcast vertex, $u\in V$ and $d_G(u,v)\leq f(v)$, then the vertex $u$ {\it  hears} a broadcast from $v$ and $v$ {\it  broadcasts to} (or {\it  $f$-dominates}) $u$. 
Note that, in particular, each vertex $v\in V^+_f$ hears a broadcast from itself
and $f$-dominates itself.

The {\it $f$-broadcast neighborhood} of a vertex $v\in V^+_f$ is the set of vertices that hear $v$, that is
$$N_f(v)=\big\{u:d_G(u,v)\leq f(v)\big\},$$
and the {\it broadcast neighborhood} of $f$ is the set 
$$N_f(V^+_f)=\cup_{v\in V^+_f} N_f(v).$$
The set of $f$-broadcast vertices that a vertex $u\in V$ can hear is the set
$$H_f(u) = \big\{v\in V^+_f: d_G(u,v)\leq f (v)\big\}.$$
For a vertex $v\in V^+_f$, the {\it  private $f$-neighborhood} of $v$ is the set of vertices that hear only $v$, that is
$$PN_f(v)=\big\{u\in V: H_f(u)=\{v\}\big\},$$
and every vertex $u\in PN_f(v)$ is a {\it private $f$-neighbor} of $v$.
Moreover,  
the {\it private $f$-border of} $v$ is either the set of private $f$-neighbors of $v$ that are at distance $f(v)$ 
from $v$, 
or the singleton $\{v\}$ if $f(v)=1$ and $PN_f(v)=\{v\}$, 
that is
$$PB_f(v) = 
\left\{
\begin{array}{ll}
    \{v\} & \text{if $f(v)=1$ and $PN_f(v)=\{v\}$},\\[2ex]
    \big\{u \in PN_f(v) : d_G(u,v) = f(v)\big\}& \text{otherwise}.
\end{array}
\right.$$
Every vertex in $PB_f(v)$ is a {\it bordering private $f$-neighbor} of $v$.
In particular, if $f(v)=1$ and $PN_f(v)=\{v\}$, then $v$ is its own bordering private $f$-neighbor.


The {\it  cost} of a broadcast $f$ on a graph $G$ is 
$$\sigma(f)=\sum_{v\in V^+_f}f (v).$$
A broadcast $f$ on $G$ of some type is {\it  minimal} (resp. {\it  maximal}) if there does not exist any broadcast $g\neq f$ on $G$ of the same type such that $g(u)\leq f(u)$ (resp. $g(u)\geq f(u)$) for all $u\in V$. 
Several types of broadcasts have been defined in the literature, in relation with domination, irredundence, independence
or packing, leading to the introduction of several broadcast numbers on graphs, corresponding to the minimum or maximum possible
cost of a maximal or minimal broadcast of the corresponding type, respectively.
For any such parameter, say $q(G)$, a broadcast $f$ on $G$ of the corresponding type with $\sigma(f)=q(G)$ will be simply called
a $q$-broadcast. We will also say that such a broadcast is {\it optimal}.

We now introduce the various types of broadcasts we will consider in this paper.


\paragraph{Dominating broadcasts.}
A broadcast $f$ on $G$ is a {\it  dominating broadcast} if every vertex in $V-V^+_f$ is $f$-dominated by some vertex in $V^+_f$ or, equivalently, if for every vertex $v\in V$, $|H_f(v)| \geq 1$. 
The {\it  broadcast domination number} $\gamma_b(G)$ of $G$ is
the minimum cost of a dominating broadcast on $G$.
The {\it  upper broadcast domination number} $\Gamma_b(G)$ of $G$
is the maximum cost of a minimal dominating broadcast on $G$.
If $f$ is a minimal dominating broadcast on $G$ such that $f(v) = 1$ for each $v\in V^+_f$, then $V^+_f$ is a {\it  minimal dominating set} in $G$, and the minimum (resp. maximum) cost of such a broadcast is the {\it  domination number $\gamma(G)$} (resp. the {\it  upper domination number $\Gamma(G)$}) of $G$.

\paragraph{Irredundant broadcasts.}
A broadcast $f$ on $G$ is an {\it  irredundant broadcast} if $PB_f(v) \neq \emptyset$ for every vertex $v\in V^+_f$. 
Stated equivalently, a broadcast $f$ is irredundant if the following two conditions are satisfied : 
(i) for every $f$-broadcast vertex $v$ with $f(v) \geq 2$, there exists a vertex $u$ such that $H_f(u) = \{v\}$ and $d_G(u,v) = f(v)$, and (ii) for every $f$-broadcast vertex $v$ with $f(v)=1$, there exists a vertex $u \in N_G[v]$ such that $H_f(u) = \{v\}$ (note that, in this case, we can have $u=v$). 
The {\it  upper broadcast irredundance number} $I\!R_b(G)$ of $G$
is the maximum cost of an irredundant broadcast on~$G$.
The {\it  broadcast irredundance number} $ir_b(G)$ of $G$
is the minimum cost of a maximal irredundant broadcast on $G$.
If $f$ is a maximal irredundant broadcast on $G$ such that $f(v) = 1$ for each $v\in V^+_f$, then $V^+_f$ is a {\it  maximal irredundant set} in $G$, and the minimum (resp. the maximum) cost of such a broadcast is the {\it  irredundance number $ir(G)$} (resp. the {\it  upper irredundance number $I\!R(G)$}) of~$G$.

\paragraph{Independent broadcasts.}
A broadcast $f$ is an {\it independent broadcast} if no broadcast vertex $f$-dominates any other broadcast vertex or, equivalently, if for every $v \in V^+_f$, $|H_f(v)| = 1$. 
The {\it broadcast independence number} $\beta_b(G)$ of $G$
is the maximum cost of an independent broadcast on $G$.
The {\it lower broadcast independence number} $i_b(G)$ of $G$
is the minimum cost of a maximal independent broadcast on $G$.
If $f$ is a maximal independent broadcast such that $f(v) = 1$ for each $v\in V^+_f$, then $V^+_f$ is a {\it  maximal independent set} in $G$, and the maximum (resp. minimum) cost of such a broadcast is the {\it  vertex independence number $\beta_0(G)$} (resp. the {\it  independent domination number $i(G)$}) of $G$.

\paragraph{Packing broadcasts.}
A broadcast $f$ is a {\it packing broadcast} if every vertex hears at most one broadcast, that is, for every vertex $v\in V$, $|H_f(v)|\leq 1$. 
The {\it broadcast packing number} $P_b(G)$ of $G$ 
is the maximum cost of a packing broadcast on $G$. 
The {\it lower broadcast packing number} $p_b(G)$ of $G$
is the minimum cost of a maximal packing broadcast on $G$. 
If $f$ is a maximal packing broadcast such that $f(v) = 1$ for each $v\in V^+_f$, then $V^+_f$ is a {\it  maximal packing set} in $G$, and the maximum (resp. the minimum) cost of such a broadcast is the {\it  packing number $P(G)$} (resp. the {\it  lower packing  number $p(G)$}) of $G$.

\medskip

\newcommand\LIGNE[4]{
\draw[thick] (#1,#2) to (#3,#4);
}
\newcommand\POINTILLE[4]{
\draw[thick,dotted] (#1,#2) to (#3,#4);
}

\newcommand\bSOM[4]{
   \node[scale=0.7,draw,circle,fill=black] at (#1,#2){};
   \node[above] at (#1,#2+0.2){#3};
   \node[below] at (#1,#2-0.2){#4};
}
\newcommand\dSOM[4]{
   \node[scale=0.7,draw,circle,fill=lightgray] at (#1,#2){};
   \node[above] at (#1,#2+0.2){#3};
   \node[below] at (#1,#2-0.2){#4};
}
\newcommand\SOM[4]{
   \node[scale=0.7,draw,circle,fill=white] at (#1,#2){};
   \node[above] at (#1,#2+0.2){#3};
   \node[below] at (#1,#2-0.2){#4};
}

\newcommand\broadcastPhuit[9]{
\foreach \k in {1,2,...,9}
   {\node[scale=0.6,draw,circle,fill=black] (v\k) at (\k,#1){};}
	
\draw[thick] (v1) to (v9);
\foreach \k in {1,2,...,9}
   {\node[below] at (\k,#1-0.2){$x_\k$};}
\node[above] at (1,#1+0.2){#2};   
\node[above] at (2,#1+0.2){#3};   
\node[above] at (3,#1+0.2){#4};   
\node[above] at (4,#1+0.2){#5};   
\node[above] at (5,#1+0.2){#6};   
\node[above] at (6,#1+0.2){#7};   
\node[above] at (7,#1+0.2){#8};   
\node[above] at (8,#1+0.2){#9};   
\node[above] at (9,#1+0.2){0};
}

\newcommand\FLECHE[2]{
   \draw[->,>=latex] (#1,#2+0.5) to[bend right=45] (#1-0.5,#2) to[bend right=45] (#1,#2-0.5);
}

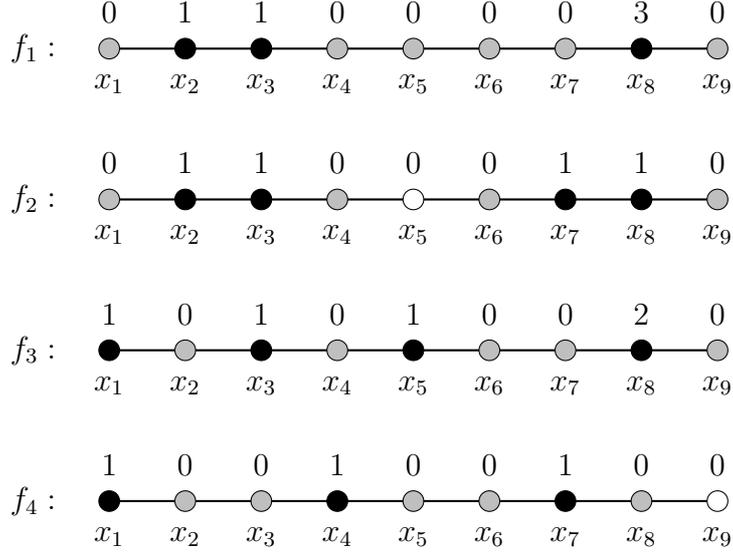
\begin{figure}
\begin{center}
\begin{tikzpicture}
\node at (0,0){$f_1:$};  
   \LIGNE{1}{0}{9}{0}
   \dSOM{1}{0}{0}{$x_1$}
   \bSOM{2}{0}{1}{$x_2$}
   \bSOM{3}{0}{1}{$x_3$}
   \dSOM{4}{0}{0}{$x_4$}
   \dSOM{5}{0}{0}{$x_5$}
   \dSOM{6}{0}{0}{$x_6$}
   \dSOM{7}{0}{0}{$x_7$}
   \bSOM{8}{0}{3}{$x_8$}
   \dSOM{9}{0}{0}{$x_9$}
\node at (0,-2){$f_2:$};  
   \LIGNE{1}{-2}{9}{-2}
   \dSOM{1}{-2}{0}{$x_1$}
   \bSOM{2}{-2}{1}{$x_2$}
   \bSOM{3}{-2}{1}{$x_3$}
   \dSOM{4}{-2}{0}{$x_4$}
   \SOM{5}{-2}{0}{$x_5$}
   \dSOM{6}{-2}{0}{$x_6$}
   \bSOM{7}{-2}{1}{$x_7$}
   \bSOM{8}{-2}{1}{$x_8$}
   \dSOM{9}{-2}{0}{$x_9$}
\node at (0,-4){$f_3:$}; 
   \LIGNE{1}{-4}{9}{-4}
   \bSOM{1}{-4}{1}{$x_1$}
   \dSOM{2}{-4}{0}{$x_2$}
   \bSOM{3}{-4}{1}{$x_3$}
   \dSOM{4}{-4}{0}{$x_4$}
   \bSOM{5}{-4}{1}{$x_5$}
   \dSOM{6}{-4}{0}{$x_6$}
   \dSOM{7}{-4}{0}{$x_7$}
   \bSOM{8}{-4}{2}{$x_8$}
   \dSOM{9}{-4}{0}{$x_9$} 
\node at (0,-6){$f_4:$};  
   \LIGNE{1}{-6}{9}{-6}
   \bSOM{1}{-6}{1}{$x_1$}
   \dSOM{2}{-6}{0}{$x_2$}
   \dSOM{3}{-6}{0}{$x_3$}
   \bSOM{4}{-6}{1}{$x_4$}
   \dSOM{5}{-6}{0}{$x_5$}
   \dSOM{6}{-6}{0}{$x_6$}
   \bSOM{7}{-6}{1}{$x_7$}
   \dSOM{8}{-6}{0}{$x_8$}
   \SOM{9}{-6}{0}{$x_9$} 
\end{tikzpicture}
\caption{\label{fig:broadcasts}Sample broadcasts on the path $P_9$.}
\end{center}
\end{figure}

%

These four different types of broadcasts are illustrated in Figure~\ref{fig:broadcasts}
(broadcast vertices are drawn as black vertices, non-broadcast dominated vertices
as gray vertices, and non dominated vertices as white vertices):
$f_1$ is a dominating broadcast, $f_2$ is an irredundant broadcast
(the $f_2$-broadcast vertices $x_2$, $x_3$, $x_7$ and $x_8$ all have a bordering private $f_2$-neighbor, namely 
$x_1$, $x_4$, $x_6$ and $x_9$, respectively),
$f_3$ is an independent broadcast, and $f_4$ is a packing broadcast.
Moreover, observe the following:
\begin{itemize}
\item $f_1$ is a {\it minimal} dominating broadcast and also a maximal irredundant broadcast.
However, $f_1$ is neither an independent broadcast (the $f_1$-broadcast vertices $x_2$ and $x_3$ are both $f_1$-dominated twice),
nor a packing broadcast ($x_2$ and $x_3$ both hear two $f_1$-broadcast vertices).
\item $f_2$ is a {\it maximal} irredundant broadcast, but is neither
a dominating broadcast ($x_5$ is not $f_2$-dominated), nor
an independent broadcast (the $f_2$-broadcast vertices $x_2$, $x_3$, $x_7$ and $x_8$ are $f_2$-dominated twice), nor
a packing broadcast ($x_2$, $x_3$, $x_7$ and $x_8$ all hear two $f_2$-broadcast vertices).
\item $f_3$ is a {\it maximal} independent broadcast and a dominating broadcast, but is neither
an irredundant broadcast (the $f_3$-broadcast vertex $x_8$ has no bordering private $f_3$-neighbor), nor
a packing broadcast (vertices $x_2$, $x_4$ and $x_6$ are $f_3$-dominated twice).
\item $f_4$ is a {\it maximal} packing broadcast, and also an irredundant broadcast and an independent broadcast.
However, $f_4$ is neither 
a dominating broadcast ($x_9$ is not $f_4$-dominated), nor
a maximal independent broadcast (we can increase the cost of $f_4$ by setting $f_4(x_1)=f_4(x_4)=f_4(x_7)=2$), nor
a maximal irredundant broadcast (we can increase the cost of $f_4$ by setting $f_4(x_7)=2$, so that $x_4$ has still
a bordering private $f_4$-neighbor, namely $x_3$, and $x_9$ is now the bordering private $f_4$-neighbor of $x_7$).
\end{itemize}

Directly from the definitions of these four types of broadcasts, we get the following observations.

\begin{observation}\label{obs:definitions}\mbox{}
\begin{itemize}
\item[{\rm (1)}] Every maximal independent broadcast is a dominating broadcast,
and thus $\gamma_b(G)\le i_b(G)$ for every graph~$G$. 
\item[{\rm (2)}] Every packing broadcast is an independent broadcast, 
and thus $P_b(G)\le \beta_b(G)$ for every graph~$G$.
\item[{\rm (3)}] \label{obs:definitions-item3} Every packing broadcast is an irredundant broadcast, 
and thus $P_b(G)\le I\!R_b(G)$ for every graph~$G$.
\item[{\rm (4)}] Every dominating maximal irredundant broadcast is a minimal dominating broadcast.
\end{itemize}
\end{observation}
%

\medskip

Broadcast domination was introduced by Erwin~\cite{Erw01} in his Ph.D. thesis, in which he discussed several types of broadcast parameters and the relationships between them.
 Many of these results appeared later in~\cite{Dun}. 
 Since then, several papers have been published on various aspects of broadcasts in graphs, including 
 the algorithmic  complexity  \cite{BR18Alg,HeLo06,HeSae12}, 
 the determination of the broadcast domination number for several classes  of graphs \cite{BHM04,BS11,Da07,H06,MW15,Se08,SK14}, 
 and a characterization of the classes of trees for which the broadcast domination number equals the radius \cite{HM09} or equals the domination number \cite{CHM10,  LM15, MW13}.  
 The upper broadcast domination number is studied in \cite{AFSHL15,BF16,Dun,Erw04,GM17,MR16},  
 the broadcast irredundance number is studied in  \cite{AFSHL15,MR16}, 
 and the broadcast independence  number is studied in \cite{A19,ABS18,ABS19,BR18Ind,BR18Inv,BouZe12}. 
 Broadcast domination and multipacking are considered in \cite{BB18,BBF18,BMY19,BMT13,HM14,MT}.

\medskip

In this paper, we determine all the above defined numbers for paths and cycles.
Ahmadi {\it et al.}  observed in~\cite{AFSHL15} that very little is known concerning these parameters.
We first recall some preliminary results in Section~\ref{sec:preliminary}, and prove our main
results in Section~\ref{sec:paths-cycles}.
These results are summarized in Table~\ref{tab:paths-cycles}.
They confirm the conjectures given in~\cite{AFSHL15} for 
$\gamma_b(P_n)$, $\gamma_b(C_n)$ and $\Gamma_b(P_n)$, but disprove all other conjectures.

\begin{table}
\centering{\small
\begin{tabular}{|c|c|c|c|c|c|c|}
\hline
     & $ \gamma_b=ir_b $ & $i_b$  & $p_b$&   $\Gamma_b=I\!R_b$ & $\beta_b$  &$P_b$  \\

\hline
   $P_n$ & $\left\lceil\frac{n}{3}\right\rceil$   
   & $\begin{array}{cc} \left\lceil\frac{2n}{5}\right\rceil,\\[1ex] n\neq 3 \end{array}$
   & $\begin{array}{cl}
 \frac{n}{4}&  \text{if }   n\equiv 0\pmod 8\\[1ex]
 2\left\lfloor\frac{n}{8}\right\rfloor +1&  \text{if }   n\equiv 1,2,3\pmod 8\\[1ex]
 2\left\lfloor\frac{n}{8}\right\rfloor +2&  \text{otherwise}\\[1ex]
                 \end{array}$ 
   & $\begin{array}{cc} n-1,\\ n\geq 2 \end{array}$
   & $\begin{array}{cc} 2n-4,\\ n\geq 3  \end{array}$
   &  $n-1$\\[4ex]
 & Th.~\ref{th:ir-gamma-path} & Th.~\ref{th:i-path} & Th.~\ref{th:p-path} & Th.~\ref{th:Gamma_IR_paths} & \cite{Erw01} & Th.~\ref{th:P-path-cycle}\\
\hline 
   $C_n$ & $\left\lceil\frac{n}{3}\right\rceil$   
   & $\begin{array}{cc} \left\lceil\frac{2n}{5}\right\rceil,\\[1ex] n\neq 3 \end{array}$ 
   &    $\begin{array}{cl}
 \frac{n}{4}&  \text{if }   n\equiv 0\pmod 8\\[1ex]
 2\left\lfloor\frac{n}{8}\right\rfloor +1&  \text{if }   n\equiv 1,2,3\pmod 8\\[1ex]
 2\left\lfloor\frac{n}{8}\right\rfloor +2&  \text{otherwise}\\[1ex]
                 \end{array}$  
   & $\begin{array}{cc} 2\big(\left\lfloor\frac{n}{2}\right\rfloor -1\big),\\[1ex] n\geq 4 \end{array}$  
   & $\begin{array}{cc} n-2,\\ n\geq 3 \end{array}$
   & $\left\lfloor\frac{n}{2}\right\rfloor$ \\[4ex]
 & Th.~\ref{th:ir-gamma-cycle} & Th.~\ref{th:i-cycle} & Th.~\ref{th:p-cycle} & Th.~\ref{th:IR-egal-Gamma-Cn}, \ref{th:Gamma-Cn} & Th.~\ref{th:beta-cycles} & Th.~\ref{th:P-path-cycle}\\
\hline
\end{tabular}
}
\caption{\label{tab:paths-cycles} Broadcast parameters of paths and cycles}
\end{table}

\section{Preliminary results}\label{sec:preliminary}

The characterization of minimal dominating broadcasts was first given by Erwin in~\cite{Erw04}, and then restated in terms of private borders\footnote{In their paper, Mynhardt and Roux used a slightly different definition of the set $PB_f(v)$ when
$f(v)=1$ and $N_f(v)\neq\{v\}$, by including the vertex $v$ in $PB_f(v)$. 
Moreover, they called the set $PB_f(v)$ the {\it private $f$-boundary} of $v$. 
We here use the term {\it private $f$-border} to avoid confusion between these two definitions.
However, it is easy to check that the private $f$-boundary of $v$ is empty if and only if the private $f$-border
of $v$ is empty, so that Proposition~\ref{prop:Erwin-PB-nonempty} is still valid in our setting.} 
by Mynhardt and Roux in~\cite{MR16}. 

\begin{proposition}[Erwin~\cite{Erw04}, restated in~\cite{MR16}]\label{prop:Erwin-PB-nonempty}  
A dominating broadcast $f$ is a minimal dominating broadcast if and only if $PB_f(v)\neq \emptyset$ for each $v\in V^+_f$.
\end{proposition}

Dunbar {\it et al.} proved in~\cite{Dun} the following bound on the upper broadcast domination number 
of graphs.

\begin{theorem}[Dunbar {\it et al.}~\cite{Dun}]\label{th:Dunbar-Gamma-m}
For every graph $G$ with size $m$, $\Gamma_b(G) \leq m$.
Moreover, $\Gamma_b(G)=m$  if and only if $G$ is a nontrivial star or path. 
\end{theorem}

This upper bound was later improved in~\cite{BF16}.

\begin{theorem}[Bouchemakh and Fergani~\cite{BF16}] \label{th:Bouchemakh-Fergani}
If $G$ is a graph of order $n$ with minimum degree $\delta(G)$, then $\Gamma_b (G) \leq n-\delta(G)$, and this bound is sharp.
\end{theorem}

From Proposition~\ref{prop:Erwin-PB-nonempty} and the  definition of a maximal irredundant broadcast, 
one gets the following result.

\begin{corollary}[Ahmadi {\it et al.}~\cite{AFSHL15}] \label{cor:Ahmadi-MinDom-is-MaxIR}
Every minimal dominating broadcast is a maximal irredundant broadcast.
\end{corollary}

Since the characteristic function of a minimal dominating set in a graph is a minimal dominating broadcast, Corollary~\ref{cor:Ahmadi-MinDom-is-MaxIR} implies the following chain of inequalities.

\begin{corollary}[Ahmadi {\it et al.}~\cite{AFSHL15}] \label{cor:Ahmadi-inequalities}
For every graph $G$, 
$$ir_b(G)\leq \gamma_b(G) \leq \gamma(G) \leq \Gamma(G) \leq \Gamma_b(G) \leq I\!R_b(G).$$
\end{corollary}

Moreover, Dunbar {\it et al.}~\cite{Dun} proved the following.

\begin{proposition}[Dunbar {\it et al.}~\cite{Dun}]\label{prop:Dunbar-inequalities}
For every graph $G$, 
$$\gamma_b(G) \leq i_b(G) \leq \beta_b(G) \geq i(G) \geq \gamma(G) \geq \gamma_b(G).$$
However, $\beta_b(G)$ and $\Gamma_b(G)$ are in general incomparable.
\end{proposition}

It is worth pointing out that the difference $I\!R_b(G)-\Gamma_b(G)$ can be arbitrarily large. 
Indeed, Mynhardt and Roux~\cite{MR16} constructed a family of graphs $\{G_r\}_{r\geq 3}$, 
where each $G_r$ is obtained by joining two copies of $K_{r+1}$ by $r$ independent edges, and proved the relation $\Gamma_b(G_r) = 3 \leq I\!R_b(G_r) = r$ for every $r \geq 3$. 
Nevertheless, we may have $I\!R_b(G)=\Gamma_b(G)$, as we will prove in Subsection~\ref{subsec:IR-Gamma} 
when $G$ is a path or a cycle.

Dunbar {\it et al.} observed in~\cite{Dun} that, for any graph $G$, neither $P(G)$ nor $p(G)$ is comparable with $p_b(G)$,
while we have $p(G)\leq P(G) \leq P_b(G)$ and   $p_b(G)\leq {\mbox{ rad}}(G)\leq {\mbox{ diam}}(G) \leq P_b(G) \leq \beta_b(G)$. 
For paths and cycles, we will prove in Section~\ref{sec:paths-cycles} that the lower bound $\diam(G)$ for $P_b(G)$ is achieved, while the difference between $\rad(G)$ and $p_b(G)$ can be arbitrarily large.

\section{Broadcast numbers of paths and cycles}\label{sec:paths-cycles}

As mentioned by Ahmadi {\it et al.} in~\cite{AFSHL15}, it is quite surprising that the values of several broadcast parameters have not been determined yet for paths or cycles. Moreover, in the same paper, they conjecture the values of these parameters. 
In this section, we will determine the exact values of these parameters, which in some cases, but not all, correspond to their
conjecture.

Throughout this section, we will denote by
$P_n=x_1x_2\dots x_n$, $n\geq 2$, the path of order $n$,
and by $C_n=x_0x_1\dots x_{n-1}$, $n\geq 3$, the cycle of order $n$.
Moreover, we assume throughout this section that
subscripts of vertices of $C_n$ are taken modulo $n$,
and that the vertices $x_1,\dots,x_n$ of $P_n$ are ``ordered''
from left to right, so that by the {\it leftmost} (resp. the {\it rightmost})
vertex in $P_n$ satisfying any property, we mean the vertex with minimum (resp. maximum)
subscript satisfying this property.

\subsection{Upper broadcast domination number and upper broadcast irredundance number}\label{subsec:IR-Gamma}

We first consider the case of paths.

\begin{theorem}\label{th:Gamma_IR_paths}
For every integer $n\geq 2$, $\Gamma_b(P_n)= I\!R_b(P_n)= \diam(P_n)=n-1$.
\end{theorem}

\begin{proof} 
Theorem~\ref{th:Dunbar-Gamma-m} directly gives $\Gamma_b(P_n)= n-1$. 
By Corollary~\ref{cor:Ahmadi-inequalities}, 
we have  $n-1=\Gamma_b(P_n)\leq I\!R_b(P_n)$. 
We now prove the opposite inequality. Let $f$ be an irredundant broadcast on $P_n$,  and let 
$V^+_f= \{x_{i_1},\dots,x_{i_t}\}$, $i_1<\cdots < i_t$, $t\geq 1$.
From the definition of an irredundant broadcast, we get that for every vertex $x_{i_j}\in V^+_f$ with $f(x_{i_j})\geq 2$, there exists a vertex $x_{i_j}^p$ such that $H_f(x_{i_j}^p) = \{x_{i_j}\}$ and $d_{P_n}(x_{i_j},x_{i_j}^p) = f(x_{i_j})$ and, for every vertex $x_{i_j}\in V^+_f$ with $f(x_{i_j})=1$, there exists a vertex $x_{i_j}^p \in N_{P_n}[x_{i_j}]$ such that $H_f(x_{i_j}^p) = \{x_{i_j}\}$. 

Let $t'$, $0\le t'\le t$, denote the number of $f$-broadcast vertices $x_{i_j}$ that are their own bordering private $f$-neighbor, 
that is, such that $x_{i_j}^p=x_{i_j}$ (which implies $f(x_{i_j})=1$), and suppose that these vertices are $\{x_{i_1},\dots,x_{i_{t'}}\}$.
We thus have
$$|V(P_n)| \geq t' + \sum_{j=t'+1}^t \big(d_{P_n}(x_{i_j},x_{i_j}^p)+1\big) 
\geq \sum_{j=1}^{t'} f(x_{i_j}) + \sum_{j=t'+1}^t \big(f(x_{i_j})+1\big) = I\!R_b(P_n) + t-t',$$  
which gives $I\!R_b(P_n)\leq n-t+t'$. 
If $t'<t$, then $I\!R_b(P_n)\leq n-1$ and we are done.
Otherwise, every $f$-broadcast vertex is its own bordering private $f$-neighbor, which implies that $V^+_f$ is
either the set of all vertices with even subscript, or the set of all vertices with odd subscript.
In both cases, we get 
$\sigma(f)=I\!R_b(P_n) \le \left\lceil\frac{n}{2}\right\rceil\le n-1$, as required.
\end{proof}

We now consider the case of cycles. For that, we will use the following lemma.

\begin{lemma}\label{lem:IR-pas-deux-H-vides}
Let $f$ be an $I\!R_b$-broadcast on $C_n$. 
If $H_f(x_{i})= \emptyset$ for some vertex $x_i$, then $H_f(x_{i-1})\neq \emptyset$ and $H_f(x_{i+1})\neq \emptyset$.
\end{lemma}

\begin{proof} 
Assume to the contrary that we have $H_f(x_{i-1})=\emptyset$, so that $x_{i-1}$ is not $f$-dominated, which implies
in particular $f(x_{i-2})=0$.
Therefore, there exists an $f$-broadcast vertex $x_j$, $j<i-2$, which $f$-dominates $x_{i-2}$, for otherwise 
we could set $f(x_{i-1})=1$, contradicting the optimality of~$f$.

We then necessarily have $d_{C_n}(x_j,x_{i-2})=f(x_j)$. 
But, in that case, the function $g$ obtained from $f$ by setting 
$g(x_j)=0$ and $g(x_{j+1})=f(x_j)+1$ would be
an irredundant broadcast on $C_n$ with cost $\sigma(g)=\sigma(f)+1 > \sigma(f)$, again a contradiction. 

It follows that we have $H_f(x_{i-1})\neq \emptyset$ and, by symmetry, that we also have $H_f(x_{i+1})\neq \emptyset$. 
\end{proof}

\begin{figure}
\begin{center}
\begin{tikzpicture}
\FLECHE{-1}{0}
\LIGNE{0}{0}{9}{0}
\POINTILLE{-1}{0}{0}{0}
\POINTILLE{9}{0}{10}{0}
   \dSOM{0}{0}{0}{}
   \dSOM{1}{0}{0}{}
   \dSOM{2}{0}{0}{}
   \bSOM{3}{0}{3}{}
   \bSOM{4}{0}{3}{{\bf 0}}
   \dSOM{5}{0}{0}{{\bf 3}}
   \dSOM{6}{0}{0}{}
   \dSOM{7}{0}{0}{}
   \SOM{8}{0}{0}{}
	\dSOM{9}{0}{0}{}
\end{tikzpicture}

\caption{\label{fig:th-IRCn} Irredundant broadcast for the proof of Theorem~\ref{th:IR-egal-Gamma-Cn}.}
\end{center}
\end{figure}
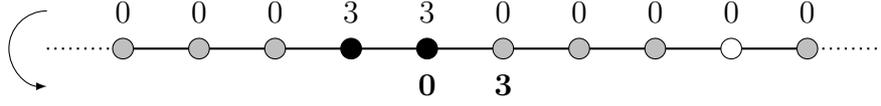

We can now prove the following result.

\begin{theorem}\label{th:IR-egal-Gamma-Cn}
For every integer $n\geq 3$, $I\!R_b(C_n)=\Gamma_b(C_n)$.
\end{theorem}

\begin{proof} 
By Corollary~\ref{cor:Ahmadi-inequalities}, we only need to prove the inequality $I\!R_b(C_n) \leq \Gamma_b(C_n)$. 
For this, it is enough to construct, from any non-dominating  $I\!R_b$-broadcast on $C_n$, 
a dominating irredundant broadcast
(which is then a minimal dominating broadcast, by Observation~\ref{obs:definitions}(4)) with the same cost $I\!R_b(C_n)$. 

Let $f$ be a non-dominating  $I\!R_b$-broadcast on $C_n$. 
Then, there exists $i\in \{0,\dots,n-1\}$ such that $H_f(x_i)=\emptyset$. 
By Lemma~\ref{lem:IR-pas-deux-H-vides}, we have $H_f(x_{i-1})\neq \emptyset$ and $H_f(x_{i+1})\neq \emptyset$. 
Since $x_i$ is not $f$-dominated, we know that  $x_{i-1}$ is $f$-dominated by a unique $f$-broadcast vertex, say $x_j$, 
$j<i-1$, such that $f(x_j)=d_{C_n}(x_j,x_{i-1})$, which implies $|PB_f(x_j)| \geq 1$.
We claim that we have $|PB_f(x_j)|=1$. 
Indeed, if $|PB_f(x_j)|=2$ , then 
we could set $f(x_{j+1})=f(x_j)$, contradicting the optimality of $f$.

%
Now, observe that 
the function $g$ obtained from $f$ by setting $g(x_j)=0$ and $g(x_{j+1})=f(x_j)$
is an irredundant broadcast with $\sigma(g) = \sigma(f)$,
such that $x_i$ is $g$-dominated,  and all vertices that were $f$-dominated remain $g$-dominated (see Figure~\ref{fig:th-IRCn})
.
Repeating the same transformation for each vertex which is not dominated, we eventually produce a minimal dominating broadcast on $C_n$ with cost $I\!R_b(C_n)$.
\end{proof}

We obviously have $\Gamma_b(C_3)=1$. 
For $n\geq 4$, the value of $\Gamma_b(C_n)$ is given by the following result.

\begin{theorem}\label{th:Gamma-Cn}
For every integer $n\geq 4$,
$\Gamma_b(C_n)= 2\big(\left\lfloor\frac{n}{2}\right\rfloor-1\big)$.
\end{theorem}

\begin{proof} 
From Theorem~\ref{th:Bouchemakh-Fergani}, we directly get $\Gamma_b(C_n)\leq n-\delta(C_n)=n-2$.
Let now $f$ be the function defined by
\[f(x_i) =\left\{\begin{array}{cl}
                 \left\lfloor\frac{n}{2}\right\rfloor -1 & \text{if } i\in \big\{\left\lfloor\frac{n}{2}\right\rfloor, \left\lceil\frac{n}{2}\right\rceil +1\big\},\\[1ex]
                 0& \text{otherwise}.
                    \end{array}
              \right.\]
Clearly, $f$ is a minimal dominating broadcast on $C_n$ with cost
\[\sigma(f)=
\left\{
\begin{array}{ll}
      n-2 & \text{if } n\text{ is even,}\\
      n-3 & \text{if } n\text{ is odd.}
\end{array}
\right. \]
Therefore, $\sigma(f) \leq \Gamma_b(C_n)$. Combining this inequality with the previous one, we already infer that we have
$\Gamma_b(C_n)= n-2$ if $n$ is even, and $n-3\leq \Gamma_b(C_n)\leq n-2$ if $n$ is odd. 

It remains to discuss the case $n$ odd.
For $n=5$, it is not difficult to check that we have $\Gamma_b(C_5)=2$. 
Suppose $n\ge 7$ and let $g$ be any  $\Gamma_b$-broadcast on $C_n$ such that 
$V^+_g=\{x_{i_1},\dots,x_{i_t}\}$, $i_1<\cdots < i_t$, $t\ge 1$.
We first prove the following claim.


\begin{claim}\label{cl:t=2}
$t=2$.
\end{claim}

\begin{proof} 
If $t=1$, then there is a unique $g$-broadcast vertex $x_i\in V^+_g$, and thus $\Gamma_b(C_n)=g(x_i) =e_{C_n}(x_i)=\frac{n-1}{2}<n-3$, a contradiction. 
Hence, $t\ge 2$.
We know by Proposition~\ref{prop:Erwin-PB-nonempty} that each $g$-broadcast vertex $x_i\in V^+_g$ has a bordering private $g$-neighbor, say $x_i^p$, with possibly $x_i^p=x_i$ (and, in that case, $g(x_i)=1$).
Let $Q_{x_i}$ be the set of edges defined as follows:
\begin{itemize}
\item if $g(x_i)\geq 2$, then $Q_{x_i}$ is the set of edges of the unique $(x_i,x_i^p)$-geodesic,
\item if $g(x_i)=1$ and $x_i^p\in\{x_{i-1},x_{i+1}\}$ is a bordering private $g$-neighbor of $x_i$, then $Q_{x_i}$ is the singleton $\{x_ix_{i}^p\}$,
\item if $g(x_i)=1$ and $x_i$ is its own bordering private $g$-neighbor, then $Q_{x_i}$ is the singleton $\{x_{i}x_{i+1}\}$.
\end{itemize}
Clearly, $g(x_i)=|Q_{x_i}|$ for every $x_i\in V^+_g$, and $Q_{x_i}\cap Q_{x_j}=\emptyset$ for every $x_i,x_j\in V^+_g$.

We now claim that for every $Q_{x_i}=\{x_ix_{i+1},\dots,x_{i+t-1}x_{i+t}\}$
(resp. $Q_{x_i}=\{x_{i-t}x_{i-t+1},\dots,x_{i-1}x_{i}\}$),
$t\ge 1$, the edge $x_{i+t}x_{i+t+1}$ (resp. $x_{i}x_{i+1}$) that ``follows'' $Q_{x_i}$
does not belong to any $Q_{x_{i'}}$, with $x_{i'}\in V^+_g$.
Indeed, this directly follows from the following observations.
\begin{itemize}
\item[$(a)$] Every $g$-broadcast vertex $x_i$ belongs to exactly one such path, namely $Q_{x_i}$.
\item[$(b)$] Every bordering private $g$-neighbor belongs to at most one such path.
\item[$(c)$] An end-vertex $x_j$ of such a path $Q$ is neither a $g$-broadcast vertex
nor a bordering private $g$-neighbor if and only if $Q=Q_{x_{j-1}}$, $x_{j-1}$ is a $g$-broadcast vertex,
$g(x_{j-1})=1$ and $x_{j-1}$ is its own bordering private $g$-neighbor.
\end{itemize}

Hence, we have
$$\Gamma_b(C_n)= \sum_{x_i\in V^+_g}g(x_i) = \sum_{x_i\in V^+_g}|Q_{x_i}| = |\cup_{x_i\in V^+_g}Q_{x_i}| \leq n-|V^+_g| = n-t.$$  

If  $t>3$, then  $\Gamma_b(C_n)< n-3$, a contradiction. 
Assume now that $t=3$ and let $V^+_g=\{x_a,x_b,x_c\}$, with $a<b<c$. 
Since $\Gamma_b(C_n) = n-3$, any two of the paths $Q_{x_a}$, $Q_{x_b}$ and $Q_{x_c}$ are separated by exactly one edge.

Suppose first that one of these $g$-broadcast vertices, say $x_c$, is such that
$g(x_c)=1$ and $x_c$ is its own private $g$-neighbor, so that $Q_{x_c} = x_cx_{c+1}$.
In that case, $x_{c-1}$ is neither a $g$-broadcast vertex nor a bordering private $g$-neighbor, which implies,
by observation $(c)$ above, that $x_b=x_{c-2}$, $g(x_b)=1$ and $x_b$ is its own bordering private $g$-neighbor.
Using the same argument, we get that $x_a=x_{b-2}$, $g(x_a)=1$, $x_a$ is its own bordering private $g$-neighbor,
and $x_c=x_{a-2}$,
leading to $n=6$, a contradiction since we assumed $n\ge 7$.

Hence, each end-vertex of any of the paths $Q_{x_a}$, $Q_{x_b}$ and $Q_{x_c}$ is either
a $g$-broadcast vertex or a bordering private $g$-neighbor of the $g$-broadcast vertex belonging to the same path.
But since we have three paths, one of $x_a$, $x_b$ or $x_c$ must be adjacent to a bordering private $g$-neighbor
of another $g$-broadcast vertex, a contradiction.
\end{proof}

By Claim~\ref{cl:t=2}, we can assume $V^+_g=\{x_a,x_b\}$, which implies $|PB_f(x_a)|=|PB_f(x_b)|$.
If $|PB_f(x_a)|=|PB_f(x_b)|=2$, then $n=2f(x_a)+1+2f(x_b)+1$, a contradiction since $n$ is odd.
We thus have $|PB_f(x_a)|=|PB_f(x_b)|=1$, and the bordering private $g$-neighbors $x_a^p$ of $x_a$, and $x_b^p$ of $x_b$
are adjacent. 
Moreover, if we assume, without loss of generality, that $x_{a}x_{a+1}\dots x_b$
is a $(x_a,x_b)$-geodesic, then $b-a\leq 2$ for otherwise the function 
$h$ obtained from $g$ by setting
$h(x_a)=0$, $h(x_b)=0$, $h(x_{a+1})=g(x_a)+1$ and $h(x_{b-1})=g(x_b)+1$,
would be a minimal dominating broadcast with cost $\Gamma_b(h)=\Gamma_b(g)+2$,
contradicting the optimality of $g$. 
Hence, we have $d_{C_n}(x_a,x_b)\leq 2$. 
If $x_a$ and $x_b$ are joined by an edge, we necessarily have $g(x_a)=g(x_b)$, implying that $n$ is even, 
contrary to our assumption.
We thus have $d_{C_n}(x_a,x_b)=2$, and thus $g(x_a)=g(x_b)=\frac{n-3}{2}$, which gives $\sigma(g)=n-3$.
\end{proof}

\subsection{Broadcast independence number and lower broadcast independence number}

Dunbar {\it et al.}~\cite{Dun} noted that the upper broadcast domination number $\Gamma_b(G)$ 
and the broadcast independence number $\beta_b(G)$ of a graph $G$ are in general incomparable.  
Erwin gave in~\cite{Erw01} the exact value of the broadcast independence number of paths. 
He proved that  for every integer $n\geq 3$, $\beta_b(P_n) = 2(n - 2)$,
so that, by Theorem~\ref{th:Gamma_IR_paths}, $\beta_b(P_n) > \Gamma_b(P_n)$ for every $n>3$.

Our next result proves that the equality 
$\beta_b(C_n)=\Gamma_b(C_n)$ holds for every cycle $C_n$, $n\ge 3$ (recall Theorem~\ref{th:Gamma-Cn}).

\begin{theorem}\label{th:beta-cycles}
For every integer $n\geq 3$,
$\beta_b(C_n)=2\big(\left\lfloor\frac{n}{2}\right\rfloor-1\big)$.
\end{theorem}

\begin{proof} 
It is easy to check that we have $\beta_b(C_3)=1$ and $\beta_b(C_4)=2$.
Assume thus $n\ge 5$. 
Clearly, the function $f$ defined by
\[f(x_i) =\left\{\begin{array}{cl}
                 \left\lfloor\frac{n}{2}\right\rfloor-1 & \text{if } i\in \big\{0,\left\lfloor\frac{n}{2}\right\rfloor\big\},\\[1ex]
                 0& \text{otherwise},
                    \end{array}
              \right.\]
is an independent broadcast on $C_n$ with cost $\sigma(f)= 2(\lfloor\frac{n}{2}\rfloor-1)$, which implies 
$\beta_b(C_n)\geq 2(\lfloor\frac{n}{2}\rfloor-1)$.

We now prove the opposite inequality. 
For this, let $g$ be any  $\beta_b$-broadcast on $C_n$.
If $|V^+_g| = 1$, say $V^+_g=\{x_i\}$,  then 
$$\sigma(g) = g(x_i) = e_{C_n}(x_i) = \left\lfloor\frac{n}{2}\right\rfloor \leq 2\Big(\left\lfloor\frac{n}{2}\right\rfloor-1\Big),$$ 
and, if $|V^+_g| = 2$, then the two vertices of $V^+_g$  are antipodal, which gives $\sigma(g)=2(\lfloor\frac{n}{2}\rfloor-1)$.  

Assume now that we have $|V^+_g|\geq 3$, and let $V^+_g = \{x_{i_0},\dots,x_{i_{k-1}}\}$, $k\ge 3$.
Since $g$ is an independent broadcast, we have $d_{C_n}(x_{i_j},x_{i_{j+1}}) \geq g(x_{i_j})+1$ for every $j$, $0\le j\le k-1$
(subscripts of $i$ are taken modulo $k$). We thus get
$$\beta_b(C_n)=\sigma(g) = \sum_{x_{i_j}\in V^+_g}g(x_{i_j}) \leq \sum_{x_{i_j}\in V^+_g}(d_{C_n}(x_{i_j},x_{i_{j+1}})-1) = n-|V^+_g| \leq n-3 \leq 2\Big(\left\lfloor\frac{n}{2}\right\rfloor-1\Big),$$
as required.
\end{proof}

We now determine the value of the lower broadcast independence number of paths and cycles.
For that, we will use the following lemma.

\begin{lemma}\label{lem:obs-i}
If $f$ is an  $i_b$-broadcast on $P_n$, $n\ge 3$,
with $V^+_f=\{x_{i_1},\dots,x_{i_t}\}$, $i_1<\cdots < i_t$, $t\ge 2$, then we have
\begin{enumerate}
\item[{\rm 1.}] $f(x_{i_1}) \geq f(x_{i_2})$ and $f(x_{i_t}) \geq f(x_{i_{t-1}})$,
\item[{\rm 2.}] $d_{P_n}(x_1,x_{i_1})\le f(x_{i_1})$ and $d_{P_n}(x_{i_t},x_n)\le f(x_{i_t})$, 
\item[{\rm 3.}] for every $j$, 
$1\le j\le t-1$, $\max\{f(x_{i_j}),f(x_{i_{j+1}})\} + 1 \leq d_{P_n}(x_{i_j},x_{i_{j+1}}) \leq f(x_{i_j}) + f(x_{i_{j+1}}) + 1$, 
\item[{\rm 4.}] $d_{P_n}(x_{i_1},x_{i_2})=f(x_{i_1})+1$ and $d_{P_n}(x_{i_{t-1}},x_{i_t}) = f(x_{i_t})+1$.
\end{enumerate}
\end{lemma}

\begin{proof}
If Item~1 is not satisfied, then we can increase by~1 the value
of $f(x_1)$ or $f(x_n)$, contradicting the maximality of $f$.

If Item~2 is not satisfied, then we can set $f(x_1)=1$, or $f(x_n)=1$, contradicting the maximality of $f$.

For Item~3, the inequality $\max\{f(x_{i_j}),f(x_{i_{j+1}})\} + 1 \leq d_{P_n}(x_{i_j},x_{i_{j+1}})$ directly follows
from the definition of an independent broadcast.
Finally, if $d_{P_n}(x_{i_j},x_{i_{j+1}}) > f(x_{i_j}) + f(x_{i_{j+1}}) + 1$,
then we can set $f(x_{i_j+f(x_{i_j})+1})=1$, contradicting the maximality of $f$.

Consider now Item~4. By items 1 and~3, we have $f(x_{i_1})+1\leq d_{P_n}(x_{i_1},x_{i_2})$
and $f(x_{i_t})+1\leq d_{P_n}(x_{i_{t-1}},x_{i_t})$.
If any of these inequalities is strict, then we can increase by~1 the $f$-value of the involved
$f$-broadcast vertex, again contradicting the maximality of $f$.
\end{proof}

In the rest of the paper, for convenience, we will often define a broadcast function $f$ on the path $P_n$ 
(resp. on the cycle $C_n$) by
the word $f(x_1)\dots f(x_n)$ (resp. $f(x_0)\dots f(x_{n-1})$), using standard notation from Formal Language Theory.
In particular, recall that when we write $(a_1\dots a_k)^q$ for some integer $q\ge 0$, we mean that the sequence $a_1\dots a_k$ is repeated exactly
$q$ times (in particular, if $q=0$, then $(a_1\dots a_k)^q$ is the empty word~$\varepsilon$).

\begin{theorem}\label{th:i-path}
For every integer $n\ge 2$, $i_b(P_n) = \left\lceil\frac{2n}{5}\right\rceil.$
\end{theorem}

\begin{proof} 
Let $n=5q+r$, with $q\ge 0$,  $0\leq r \leq 4$ and $5q+r\ge 2$.
It is easy to check that the functions $10$, $010$ and $0101$ are  $i_b$-broadcasts on $P_n$ 
with cost $\left\lceil\frac{2n}{5}\right\rceil$ when $n=2$, $3$ and $4$, respectively.

Assume now $n\geq 5$. 
According to the value of $r$,  we define the broadcasts $f_r$ on $P_n$ for each $r$, $0\le r\le 4$, as follows:
$$f_0(P_n) = (01010)^q,\ f_1(P_n) = (01010)^q1,\ f_2(P_n) = (01010)^q10,$$
$$f_3(P_n) = (01010)^q101,\ \mbox{and}\ f_4(P_n) = (01010)^q0101.$$

It is not difficult to check that each $f_r$, $0\le r\le 4$, is a maximal independent broadcast 
on $C_{5q+r}$, for each $q\ge 1$, with cost $\sigma(f_r)=\left\lceil\frac{2n}{5}\right\rceil$, which gives
$i_b(P_n)\leq \lceil\frac{2n}{5}\rceil$. 

\medskip

Let us now consider the opposite inequality. 
Let $f$ be an  $i_b$-broadcast on~$P_n$. 
If $|V^+_f|=1$, then we have $i_b(P_n)=rad(P_n) = \left\lfloor\frac{n}{2}\right\rfloor$,
which implies $n\in S=\{1,\ldots,9,11,13\}$, since otherwise we would have 
$i_b(P_n)=\left\lfloor\frac{n}{2}\right\rfloor > \left\lceil\frac{2n}{5}\right\rceil$,
contradicting the inequality we have established before. 
Observe also that we have $\left\lfloor\frac{n}{2}\right\rfloor = \left\lceil\frac{2n}{5}\right\rceil$
for every $n\in S$, so that we are done.

The remaining case we have to consider is thus $n\notin S$, which implies $|V^+_f|\ge 2$.
Let $V^+_f=\{x_{i_1},\dots,x_{i_t}\}$, $i_1<\cdots < i_t$, with $t\ge 2$.
The two following claims will prove that we can always choose $f$ such that $f(x_{i_j})=1$ for every
$f$-broadcast vertex $v_{i_j}\in V^+_f$.

\begin{claim}\label{cl:que-des-uns-bords}
There exists an  $i_b$-broadcast $g$ on $P_n$ such that
$g(x_{i_1})=g(x_{i_2})=g(x_{i_{t-1}})=g(x_{i_t})=1$.
\end{claim}

\begin{proof}
We first prove that there exists an  $i_b$-broadcast $g_0$ on $P_n$ such that
$g_0(x_{i_1})=g_0(x_{i_2})=1$.
If $f(x_{i_1})=f(x_{i_2})=1$, then we set $g_0:=f$ and we are done.
So, suppose that we have $f(x_{i_1}) + f(x_{i_2}) \geq 3$.

%

If $|V^+_f|=2$, then, by Lemma~\ref{lem:obs-i}(4), we have $d_{P_n}(x_{i_1},x_{i_2})=f(x_{i_1})+1=f(x_{i_2})+1$.
Using Lemma~\ref{lem:obs-i}(2) and~(3), we then get
$$n = d_{P_n}(x_1,x_{i_1}) + d_{P_n}(x_{i_1},x_{i_2}) + d_{P_n}(x_{i_2},x_n) + 1 \leq 3f(x_{i_1}) + 2,$$
and thus
$\sigma(f)=2f(x_{i_1})\geq \frac{2(n-2)}{3}$.
Now, recall that the above defined maximal independent broadcast $f_r$, with $r=n\mod 5$,
is such that $\sigma(f_r)=\left\lceil\frac{2n}{5}\right\rceil$.
Since $n\ge 10$, this contradicts the optimality of $f$.
%
%

We thus have $|V^+_f|\ge 3$.
We consider two cases, depending on the value of $d_{P_n}(x_{i_2},x_{i_3})$.

\begin{figure}
\begin{center}
\begin{tikzpicture}
\FLECHE{-1}{0}
\LIGNE{0}{0}{8}{0}
\POINTILLE{8}{0}{9}{0}
   \dSOM{0}{0}{0}{{\bf 0}}
   \dSOM{1}{0}{0}{{\bf 1}}
   \dSOM{2}{0}{0}{{\bf 0}}
   \bSOM{3}{0}{3}{{\bf 1}}
   \dSOM{4}{0}{0}{{\bf 0}}
   \dSOM{5}{0}{0}{{\bf 1}}
   \dSOM{6}{0}{0}{{\bf 0}}
   \bSOM{7}{0}{1}{{\bf 1}}
   \dSOM{8}{0}{0}{}
\end{tikzpicture}

(a) $i_2$ is even ($8$ in this example)

\vskip 0.8cm

\begin{tikzpicture}
\FLECHE{-1}{0}
\LIGNE{0}{0}{7}{0}
\POINTILLE{7}{0}{8}{0}
   \dSOM{0}{0}{0}{{\bf 1}}
   \dSOM{1}{0}{0}{{\bf 0}}
   \bSOM{2}{0}{3}{{\bf 1}}
   \dSOM{3}{0}{0}{{\bf 0}}
   \dSOM{4}{0}{0}{{\bf 1}}
   \dSOM{5}{0}{0}{{\bf 0}}
   \bSOM{6}{0}{1}{{\bf 1}}
   \dSOM{7}{0}{0}{}
\end{tikzpicture}

(b) $i_2$ is odd ($7$ in this example)

\caption{\label{fig:th-ibPn-b} Maximal independent broadcast for the proof of Claim~\ref{cl:que-des-uns-bords}, 
Case 1.}
\end{center}
\end{figure}

\begin{enumerate}
\item $f(x_{i_3})+1 \le d_{P_n}(x_{i_2},x_{i_3}) \le f(x_{i_3})+2$.\\
Let $g$ be the mapping obtained from $f$ by replacing the $f$-values $0^{i_1-1}f(x_{i_1})0^{i_2-i_1-1}f(x_{i_2})$
of $x_1\dots x_{i_1}\dots x_{i_2}$ by $(01)^{\frac{i_2}{2}}$ if $i_2$ is even, or by $1(01)^{\frac{i_2-1}{2}}$ if $i_2$ is odd 
(see Figure~\ref{fig:th-ibPn-b}).
In both cases, we have $g(x_{i_2})=1$, which implies that $g$ is a maximal independent broadcast on $P_n$.
We then have
$$\sigma(g) - \sigma(f) = \left\lceil\frac{i_2}{2}\right\rceil - f(x_{i_1}) - f(x_{i_2}).$$
By Lemma~\ref{lem:obs-i}(2) and~(4), we have
$$i_2 = d_{P_n}(x_1,x_{i_1}) + d_{P_n}(x_{i_1},x_{i_2}) + 1 \leq 2f(x_{i_1}) + 2,$$
and thus
$$\sigma(g) - \sigma(f) \leq f(x_{i_1}) + 1 - f(x_{i_1}) - f(x_{i_2}) = 1 - f(x_{i_2}).$$
The optimality of $f$ then implies $f(x_{i_2})=1$, so that we have $\sigma(g) = \sigma(f)$
and we can set $g_0:=g$.

\begin{figure}
\begin{center}
\begin{tikzpicture}
\FLECHE{-1}{0}
\LIGNE{0}{0}{12}{0}
\POINTILLE{12}{0}{13}{0}
   \dSOM{0}{0}{0}{{\bf 0}}
   \dSOM{1}{0}{0}{{\bf 1}}
   \bSOM{2}{0}{3}{{\bf 0}}
   \dSOM{3}{0}{0}{{\bf 1}}
   \dSOM{4}{0}{0}{{\bf 0}}
   \dSOM{5}{0}{0}{{\bf 1}}
   \bSOM{6}{0}{3}{{\bf 0}}
   \dSOM{7}{0}{0}{{\bf 1}}
   \dSOM{8}{0}{0}{}
	\dSOM{9}{0}{0}{}
	\dSOM{10}{0}{0}{}
	\bSOM{11}{0}{2}{}
	\dSOM{12}{0}{0}{}
\end{tikzpicture}

(a) $i_3-f(x_{i_3})-2$ is even ($12-2-2=8$ in this example)

\vskip 0.8cm

\begin{tikzpicture}
\FLECHE{0}{0}
\LIGNE{1}{0}{12}{0}
\POINTILLE{12}{0}{13}{0}
   \dSOM{1}{0}{0}{{\bf 1}}
   \bSOM{2}{0}{3}{{\bf 0}}
   \dSOM{3}{0}{0}{{\bf 1}}
   \dSOM{4}{0}{0}{{\bf 0}}
   \dSOM{5}{0}{0}{{\bf 1}}
   \bSOM{6}{0}{3}{{\bf 0}}
   \dSOM{7}{0}{0}{{\bf 1}}
   \dSOM{8}{0}{0}{}
	\dSOM{9}{0}{0}{}
	\dSOM{10}{0}{0}{}
	\bSOM{11}{0}{2}{}
	\dSOM{12}{0}{0}{}
\end{tikzpicture}

(b) $i_3-f(x_{i_3})-2$ is odd ($11-2-2=7$ in this example)

\caption{\label{fig:th-ibPn-b-1} Maximal independent broadcast for the proof of Claim~\ref{cl:que-des-uns-bords}, Case 2.}
\end{center}
\end{figure}

\item $d_{P_n}(x_{i_2},x_{i_3}) \ge f(x_{i_3})+3$.\\
%
%
Let $g$ be the mapping obtained from $f$ by replacing the $f$-values 
$$0^{i_1-1}f(x_{i_1})0^{i_2-i_1-1}f(x_{i_2}) 0^{i_3-i_2-f(x_{i_3})-2}$$
of $x_1\dots x_{i_1}\dots x_{i_2} \dots x_{i_3-f(x_{i_3})-2}$ 
by $(01)^{\frac{i_3-f(x_{i_3})-2}{2}}$ if $i_3-f(x_{i_3})-2$ is even, or by $1(01)^{\frac{i_3-f(x_{i_3})-3}{2}}$ if $i_3-f(x_{i_3})-2$ is odd (see Figure~\ref{fig:th-ibPn-b-1}).

In both cases, all vertices are $g$-dominated and 
$g$ is clearly a maximal independent broadcast on $P_n$.
We then have
$$\sigma(g) - \sigma(f) = \left\lceil\frac{i_3-f(x_{i_3})-2}{2}\right\rceil - f(x_{i_1}) - f(x_{i_2}).$$
By Lemma~\ref{lem:obs-i}(2) and~(4), we have
$$i_2 = d_{P_n}(x_1,x_{i_1}) + d_{P_n}(x_{i_1},x_{i_2}) + 1 \leq 2f(x_{i_1}) + 2,$$
and thus $i_3-f(x_{i_3})-2 \leq 2f(x_{i_1})+f(x_{i_2})+1$.
We then get  
$$\sigma(g) - \sigma(f) \leq f(x_{i_1}) + \left\lceil \frac{f(x_{i_2})+1}{2}\right\rceil - f(x_{i_1}) - f(x_{i_2}) = \left\lceil \frac{1 - f(x_{i_2})}{2}\right\rceil.$$
The optimality of $f$ then implies $f(x_{i_2})\le 2$, so that we have $\sigma(g) = \sigma(f)$
and we can set $g_0:=g$.
\end{enumerate}

\medskip

Observe now that in both of the above cases we have $g_0(x_{i_j}) = f(x_{i_j})$
for every $j$, $3\le j\le t$. Therefore, using symmetry and starting from $g_0$ instead of $f$, we
can similarly construct an  $i_b$-broadcast $g$ on $P_n$ with
$g(x_{i_1})=g(x_{i_2})=g(x_{i_{t-1}})=g(x_{i_t})=1$, as required.
\end{proof}

\begin{claim}\label{cl:que-des-uns}
There exists an  $i_b$-broadcast $g$ on $P_n$ such that
$g(x_i)=1$ for every vertex $x_i\in V_g^+$.
\end{claim}

\begin{proof}
By Claim~\ref{cl:que-des-uns-bords}, we can suppose that 
$f(x_{i_1})=f(x_{i_2})=f(x_{i_{t-1}})=f(x_{i_t})=1$.
If $f(x_i)=1$ for every vertex $x_i\in V^+_f$, there is nothing to prove.
Suppose thus that this is not the case, which implies $t\ge 5$, and let
$x_{i_j}$, $3\le j\le t-2$, be the leftmost $f$-broadcast vertex for which $f(x_{i_j})\geq 2$.

We will prove that we can always construct an  $i_b$-broadcast
$f'$ on $P_n$ such that the number of broadcast vertices with $f'$-value at least~2 is strictly
less than the number of broadcast vertices with $f$-value at least~2.


We consider two cases, depending on the value of $d_{P_n}(x_{i_{j+1}},x_{i_{j+2}})$.

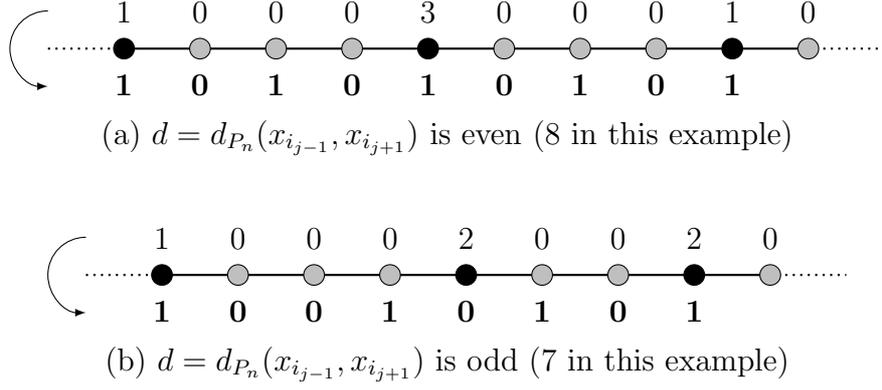
\begin{figure}
\begin{center}
\begin{tikzpicture}
\FLECHE{-1}{0}
\LIGNE{0}{0}{9}{0}
\POINTILLE{-1}{0}{0}{0}
\POINTILLE{9}{0}{10}{0}
   \bSOM{0}{0}{1}{{\bf 1}}
   \dSOM{1}{0}{0}{{\bf 0}}
   \dSOM{2}{0}{0}{{\bf 1}}
   \dSOM{3}{0}{0}{{\bf 0}}
   \bSOM{4}{0}{3}{{\bf 1}}
   \dSOM{5}{0}{0}{{\bf 0}}
   \dSOM{6}{0}{0}{{\bf 1}}
   \dSOM{7}{0}{0}{{\bf 0}}
   \bSOM{8}{0}{1}{{\bf 1}}
   \dSOM{9}{0}{0}{}
\end{tikzpicture}

(a) $d=d_{P_n}(x_{i_{j-1}},x_{i_{j+1}})$ is even ($8$ in this example)

\vskip 0.8cm

\begin{tikzpicture}
\FLECHE{-1}{0}
\LIGNE{0}{0}{8}{0}
\POINTILLE{-1}{0}{0}{0}
\POINTILLE{8}{0}{9}{0}
   \bSOM{0}{0}{1}{{\bf 1}}
   \dSOM{1}{0}{0}{{\bf 0}}
   \dSOM{2}{0}{0}{{\bf 0}}
   \dSOM{3}{0}{0}{{\bf 1}}
   \bSOM{4}{0}{2}{{\bf 0}}
   \dSOM{5}{0}{0}{{\bf 1}}
   \dSOM{6}{0}{0}{{\bf 0}}
   \bSOM{7}{0}{2}{{\bf 1}}
   \dSOM{8}{0}{0}{}
	
\end{tikzpicture}

(b) $d = d_{P_n}(x_{i_{j-1}},x_{i_{j+1}})$ is odd ($7$ in this example)

\caption{\label{fig:th-ibPn-c} Maximal independent broadcast for the proof of Claim~\ref{cl:que-des-uns}, Case 1.}
\end{center}
\end{figure}

\begin{enumerate}
\item $f(x_{i_{j+2}})+1 \le d_{P_n}(x_{i_{j+1}},x_{i_{j+2}}) \le f(x_{i_{j+2}})+2$.\\
Let
$d = i_{j+1}-i_{j-1} = d_{P_n}(x_{i_{j-1}},x_{i_{j+1}}) = d_{P_n}(x_{i_{j-1}},x_{i_{j}}) + d_{P_n}(x_{i_{j}},x_{i_{j+1}})$,
and $f'$ be the mapping obtained from $f$ by replacing the $f$-values 
$$f(x_{i_{j-1}})0^{i_j-i_{j-1}-1}f(x_{i_j})0^{i_{j+1}-i_j-1}f(x_{i_{j+1}})$$
of $x_{i_{j-1}}\dots x_{i_j}\dots x_{i_{j+1}}$ by 
$1(01)^{\frac{d}{2}}$
if $d$ is even, or by 
$10(01)^{\frac{d-1}{2}}$
if $d$ is odd 
(see Figure~\ref{fig:th-ibPn-c}).


%
Observe that we have $f'(x_{i_{j-1}})=1$ and $f'(x_{i_{j+1}})=1$
in both cases, which implies that $f'$ is a maximal independent broadcast on $P_n$.
Moreover, in both cases, we have
$$\sigma(f')-\sigma(f) = \left\lfloor\frac{d}{2}\right\rfloor - f(x_{i_{j}}) - f(x_{i_{j+1}}).$$

Since $f(x_{i_{j-1}})=1 < f(x_{i_j})$, 
Lemma~\ref{lem:obs-i}(3) gives
$$f(x_{i_j}) + 1 \leq d_{P_n}(x_{i_{j-1}},x_{i_{j}}) \leq f(x_{i_j}) + 2.$$

We thus consider two subcases, depending on the value of $d_{P_n}(x_{i_{j-1}},x_{i_{j}})$.

\begin{enumerate}
\item $d_{P_n}(x_{i_{j-1}},x_{i_{j}}) = f(x_{i_j}) + 1$.\\
By Lemma~\ref{lem:obs-i}(3), we have
$$d = d_{P_n}(x_{i_{j-1}},x_{i_{j}}) + d_{P_n}(x_{i_{j}},x_{i_{j+1}}) \leq  f(x_{i_{j}}) + 1 + f(x_{i_{j}}) + f(x_{i_{j+1}}) + 1
= 2f(x_{i_{j}}) + f(x_{i_{j+1}}) + 2,$$
and thus
$$\sigma(f')-\sigma(f) \leq  f(x_{i_{j}}) + \left\lfloor\frac{f(x_{i_{j+1}})}{2}\right\rfloor + 1 - f(x_{i_{j}}) - f(x_{i_{j+1}})
= \left\lfloor\frac{2-f(x_{i_{j+1}})}{2}\right\rfloor \leq 0.$$
The optimality of $f$ then implies the optimality of $f'$, 
and the number of broadcast vertices with $f'$-value at least~2 is strictly
less than the number of broadcast vertices with $f$-value at least~2, as required.  

\item $d_{P_n}(x_{i_{j-1}},x_{i_{j}}) = f(x_{i_j}) + 2$.\\
Since $f$ is maximal, we necessarily have $d_{P_n}(x_{i_{j}},x_{i_{j+1}}) = f(x_{i_j}) + 1$,
since otherwise we could increase $f(x_{i_j})$ by~1.
This gives
$$d = d_{P_n}(x_{i_{j-1}},x_{i_{j}}) + d_{P_n}(x_{i_{j}},x_{i_{j+1}}) = 2f(x_{i_j}) + 3,$$
and thus
$$\sigma(f')-\sigma(f) = f(x_{i_{j}}) + 1 - f(x_{i_{j}}) - f(x_{i_{j+1}})
= 1 - f(x_{i_{j+1}}) \leq 0.$$
Again, the optimality of $f$ then implies the optimality of $f'$, 
and the number of broadcast vertices with $f'$-value at least~2 is strictly
less than the number of broadcast vertices with $f$-value at least~2, as required.  
\end{enumerate}

\begin{figure}
\begin{center}

\begin{tikzpicture}
\FLECHE{-1}{0}
\LIGNE{0}{0}{16}{0}
\POINTILLE{-1}{0}{0}{0}
\POINTILLE{16}{0}{17}{0}
   \bSOM{0}{0}{1}{{\bf 1}}
   \dSOM{1}{0}{0}{{\bf 0}}
   \dSOM{2}{0}{0}{{\bf 0}}
   \bSOM{3}{0}{2}{{\bf 1}}
   \dSOM{4}{0}{0}{{\bf 0}}
   \dSOM{5}{0}{0}{{\bf 1}}
   \dSOM{6}{0}{0}{{\bf 0}}
   \bSOM{7}{0}{3}{{\bf 1}}
   \dSOM{8}{0}{0}{{\bf 0}}
	\dSOM{9}{0}{0}{\bf{1}}
	\dSOM{10}{0}{0}{}
	\dSOM{11}{0}{0}{}
	\dSOM{12}{0}{0}{}
	\bSOM{13}{0}{2}{}
\dSOM{14}{0}{0}{}
\dSOM{15}{0}{0}{}
\bSOM{16}{0}{}{}
\end{tikzpicture}

(a) $d'=i_{j+2}-i_{j-1}-f(x_{i_{j+2}})-3$ is even ($13-2-3=8$ in this example)

\vskip 0.8cm

\begin{tikzpicture}
\FLECHE{-1}{0}
\LIGNE{0}{0}{15}{0}
\POINTILLE{-1}{0}{0}{0}
\POINTILLE{15}{0}{16}{0}
   \bSOM{0}{0}{1}{{\bf 1}}
   \dSOM{1}{0}{0}{{\bf 0}}
   \dSOM{2}{0}{0}{{\bf 1}}
   \bSOM{3}{0}{2}{{\bf 0}}
   \dSOM{4}{0}{0}{{\bf 1}}
   \dSOM{5}{0}{0}{{\bf 0}}
   \dSOM{6}{0}{0}{{\bf 1}}
   \bSOM{7}{0}{3}{{\bf 0}}
   \dSOM{8}{0}{0}{1}
	\dSOM{9}{0}{0}{}
	\dSOM{10}{0}{0}{}
	\dSOM{11}{0}{0}{}
	\bSOM{12}{0}{2}{}
	\dSOM{13}{0}{0}{}
\dSOM{14}{0}{0}{}
\bSOM{15}{0}{}{}
\end{tikzpicture}

(b) $d'=i_{j+2}-i_{j-1}-f(x_{i_{j+2}})-3$ is odd ($12-2-3=7$ in this example)

\caption{\label{fig:th-ibPn-c-1} Maximal independent broadcast for the proof of Claim~\ref{cl:que-des-uns}, Case 2.}
\end{center}
\end{figure}

\item $d_{P_n}(x_{i_{j+1}},x_{i_{j+2}}) \geq f(x_{i_{j+2}})+3$.\\
Let $d'= d_{P_n}( x_{i_{j+2}-f( x_{i_{j+2}} )-2}, x_{i_{j-1}+1}) =
i_{j+2}-f(x_{i_{j+2}})-2-i_{j-1}-1=i_{j+2}-i_{j-1}-f(x_{i_{j+2}})-3$,
and $f'$ be the mapping obtained from $f$ by replacing the $f$-values 
$$f(x_{i_{j-1}})0^{i_j-i_{j-1}-1}f(x_{i_j})0^{i_{j+1}-i_j-1}f(x_{i_{j+1}}) 0^{i_{j+2}-f(x_{i_{j+2}})-2 - i_{j+1} }$$
of $x_{i_{j-1}}\dots x_{i_j}\dots x_{i_{j+1}} \dots x_{i_{j+2}-f(x_{i_{j+2}})-2} $ by 
$10(01)^{\frac{d'}{2}}$
if $d'$ is even, or by 
$1(01)^{\frac{d'+1}{2}}$
if $d'$ is odd (see Figure~\ref{fig:th-ibPn-c-1}).


Since $i_{j+2}-i_{j+1} \leq f(x_{i_{j+1}})+f(x_{i_{j+2}})+1$
and $i_{j+1}-i_{j-1} \leq 2f(x_{i_j})+f(x_{i_{j+1}})+2$,
we get 
$i_{j+2}-i_{j-1} \leq 2\ f(x_{i_j})+2f(x_{i_{j+1}}) +f(x_{i_{j+2}})+3$ and thus 
$d'\leq 2f(x_{i_j})+2f(x_{i_{j+1}})$.

Hence, the mapping $f'$ is a maximal independent broadcast on $P_n$ and
we have
$$\sigma(f')-\sigma(f) = \left\lceil\frac{d'}{2}\right\rceil - f(x_{i_{j}}) - f(x_{i_{j+1}}) 
\le \left\lceil\frac{2f(x_{i_j})+2f(x_{i_{j+1}})}{2}\right\rceil - f(x_{i_{j}}) - f(x_{i_{j+1}}) = 0.$$
The optimality of $f$ then implies the optimality of $f'$, 
and the number of broadcast vertices with $f'$-value at least~2 is strictly
less than the number of broadcast vertices with $f$-value at least~2, as required.  
\end{enumerate}

In each case, we were able to construct an  $i_b$-broadcast $f'$
such that the number of broadcast vertices with $f'$-value at least~2 is strictly
less than the number of broadcast vertices with $f$-value at least~2, as required. 
Repeating this construction until no such broadcast vertex exists, we
eventually get an  $i_b$-broadcast $g$
such that $g(x_{i_j})=1$ for every $g$-broadcast vertex $x_{i_j}$.
This completes the proof of Claim~\ref{cl:que-des-uns}.
\end{proof}

By Claim~\ref{cl:que-des-uns}, we can thus now assume that the 
 $i_b$-broadcast $f$ is such that $f(x_{i_j})=1$ for every
$f$-broadcast vertex $x_{i_j}$.
It remains to prove that, for every $n\ge 5$, $\sigma(f)=\left\lceil\frac{2n}{5}\right\rceil$.

For that, we first prove the following claim. Let $w_f$ denote the word on the alphabet $\{0,1\}$
defined by $w_f = f(x_1)\dots f(x_n)$.

\begin{claim}\label{cl:2n-sur-5}
There exists an
 $i_b$-broadcast $f$  on $P_n$, $n\geq 5$, such that $f(x_{i_j})=1$ for every
$f$-broadcast vertex $x_{i_j}$, that satisfies the following properties:
\begin{enumerate}
\item[{\rm (1)}] $w_f$ does not contain the factor $000$, and
\item[{\rm (2)}] $w_f$ does not contain the factor $1001001$.
\item[{\rm (3)}] $0101$ is a prefix of $w_f$,
\item[{\rm (4)}] either $101$ or $1010$ is a suffix of $w_f$.
\end{enumerate}
\end{claim}

\begin{proof}
If $f(x_{i-1})f(x_{i})f(x_{i+1})=000$, then we can set $f(x_{i})=1$, 
contradicting the maximality of $f$, which proves Item~(1).
Similarly, if $f(x_{i-3})\dots f(x_{i})\dots f(x_{i+3})=1001001$, then we can set $f(x_{i})=2$, 
again contradicting the maximality of $f$, which proves Item~(2).


Now, observe that we can have neither
$f(x_1)f(x_2) = 00$, since otherwise we could set $f(x_1)=1$, neither 
$f(x_1)\dots f(x_4) = 0100$, since otherwise we could set $f(x_2)=2$, nor 
$f(x_1)\dots f(x_4) = 1001$, since otherwise we could set $f(x_1)=2$,  
contradicting in each case the maximality of $f$.
Therefore, either $0101$ or $1010$ is a prefix of $w_f$.
In the former case we are done, so let us assume that $1010$ is a prefix of $w_f$.
Suppose that $w_f$ contains the factor $100$ and consider its first occurrence, that is,
suppose that $(10)^k100$ is a prefix of $w_f$ for some $k\ge 1$.
In that case, we can replace the $f$-values $(10)^k100$ of $x_1\dots x_{2k+3}$ by $0(10)^k10$
and we are done.
If $w_f$ does not contain the factor $100$, then we necessarily have either
$w_f=(10)^{\frac{n}{2}}$ or $w_f=(10)^{\frac{n-1}{2}}1$, which gives 
$\sigma(f)=\left\lceil\frac{n}{2}\right\rceil > \left\lceil\frac{2n}{5}\right\rceil$, 
a contradiction.
This proves Item~(3).

By symmetry, using the same argument as for the previous item,
we get that none of $00$, $0010$, or $1001$ can be a suffix of $w_f$.
Hence, either $1010$ or $0101$ is a suffix of $w_f$,
which proves Item~(4).
\end{proof}

By Item~(3) of Claim~\ref{cl:2n-sur-5}, we let $w'_f$ be the
word defined by $w_f=0101w'_f$.
We will now ``split'' $w'_f$ in factors (or blocks) $w_1,\dots, w_q$, $q\ge 0$
($q=0$ meaning that no such block appeared),
each of length 2 or~5,
inductively defined as follows.
\begin{itemize}
\item If $01$ is a prefix of $w'_f$, then $w_1=01$,
\item If $00101$ is a prefix of $w'_f$, then $w_1=00101$,
\item If $w'_f=w_1\dots w_{k-1}w''$ and $01$ is a prefix of $w''$, then $w_k=01$,
\item If $w'_f=w_1\dots w_{k-1}w''$ and $00101$ is a prefix of $w''$, then $w_k=00101$.
\end{itemize}
We then have either $w_f=0101w''$ and $q=0$, or $w_f=0101w_1\dots w_qw''$, for some $q\ge 1$, 
with $w''$ being either empty or 0
(observe that $w''$ cannot start with a 1, and recall that 00 cannot be a suffix of $w_f$),
and, if $q>0$, then $w_i\in\{01,00101\}$ for every $i$, $1\le i\le q$.

Observe that if $w_i=01$ and $w_{i+1}=00101$ for some $i$, $1\le i\le q-1$,
then the mapping $g$ defined by $w_g=0101w_1\dots w_{i-1}.00101.01.w_{i+2}\dots w''$
is still an  $i_b$-broadcast on $P_n$.
Therefore, $f$ can be chosen in such a way that there exists some $q_0\le q$
such that $w_i=00101$ if and only if $i\le q_0$.

We now claim that $f$ can be chosen in such a way that we have at most two blocks equal to 01, 
that is, $q-q_0\le 2$.
Indeed, if we add at least three such blocks, then either $1010101$ or $10101010$
is a suffix of $w_f$. In the former case, $1010101$ could be replaced by $1001010$,
contradicting the optimality of $f$.
In the latter case, we may replace $10101010$ by $10010101$, so that $q-q_0=2$.

Finally, we get that the structure of $w_f$ (recall that $n\ge 5$) is either 
$$01010,\ 010101,\ 0101010,\ 01010101,\ \text{or}\ 0101w_1\dots w_{q_0}w',$$
with ${q_0}\ge 1$, $w_i=00101$ for every $i$, $1\le i\le {q_0}$,
and $w'\in\{\varepsilon,0,01,010,0101,01010\}$.
It is now routine to check that $\sigma(f)=\left\lceil\frac{2n}{5}\right\rceil$ in each case,
which gives $i_b(P_n) = \left\lceil\frac{2n}{5}\right\rceil$ for every $n\ge 3$.

This completes the proof.
\end{proof}

Using Theorem~\ref{th:i-path}, we can also prove a similar result for cycles.


\begin{theorem}\label{th:i-cycle}
For every integer $n\ge 3$, $i_b(C_n)= \left\lceil\frac{2n}{5}\right\rceil.$
\end{theorem}

\begin{proof}
Observe first that 010 and 0101 are   $i_b$-broadcasts on $C_n$
with cost $\left\lceil\frac{2n}{5}\right\rceil$, when $n=3,4$, respectively,
and that the five functions $f_0,\dots,f_4$, defined in the proof of Theorem~\ref{th:i-path},
are also   $i_b$-broadcasts on $C_n$, $n\ge 5$, with cost $\left\lceil\frac{2n}{5}\right\rceil$.
We thus have $i_b(C_n) = \left\lceil\frac{2n}{5}\right\rceil$ for $n=3,4$, and
$i_b(C_n) \leq \left\lceil\frac{2n}{5}\right\rceil$ for every $n\geq 5$.

We now prove the opposite inequality when $n\ge 5$.
For that, let $f$ be any   $i_b$-broadcast on $C_n$, $n\ge 5$.
Suppose first that $|V^+_f|=1$, which implies $i_b(C_n)=\diam(C_n)=\left\lfloor\frac{n}{2}\right\rfloor$.
As observed in the proof of Theorem~\ref{th:i-path}, according to the inequality
we established before, this situation can only happen if $n\in S=\{1,\dots,9,11,13\}$.

Suppose now that $n\notin S$, which implies $|V^+_f| \geq 2$ and, in particular, $n\ge 10$.
We now claim that we necessarily have $|V^+_f| \geq 3$.
Indeed, suppose to the contrary that $|V^+_f| = 2$, and let $V^+_f=\{x_{i_1},x_{i_2}\}$.
Denote by $Q_1=x_{i_1}x_{i_1+1}\dots x_{i_2}$, and $Q_2=x_{i_2}x_{i_2+1}\dots x_{i_1}$
the two paths joining $x_{i_1}$ and $x_{i_2}$.
Since $f$ is maximal, we can assume, without loss of generality,
that $|Q_1|=d_{C_n}(x_{i_1},x_{i_2})=f(x_{i_1})+1=f(x_{i_2})+1$, which gives $f(x_{i_1})=f(x_{i_2})$.
Since Item~3 of Lemma~\ref{lem:obs-i} also holds for cycles, we have 
$|Q_2| \le f(x_{i_1})+f(x_{i_2})+1 = 2f(x_{i_1}) + 1$, which gives
$$n = |Q_1| + |Q_2| \leq f(x_{i_1})+1 + 2f(x_{i_1}) + 1 = 3f(x_{i_1}) + 2.$$
We then get $f(x_{i_1}) \geq \left\lceil\frac{n-2}{3}\right\rceil$, and thus
$$i_b(C_n) = 2f(x_{i_1}) \geq 2\cdot\left\lceil\frac{n-2}{3}\right\rceil,$$
a contradiction with the inequality we established before since 
$2\left\lceil\frac{n-2}{3}\right\rceil > \left\lceil\frac{2n}{5}\right\rceil$
when $n\ge 10$.

We can thus assume $|V^+_f|\geq 3$, and let $V^+_f = \{x_{i_0},\dots,x_{i_{t-1}}\}$, $t\ge 3$.
We first claim that there exists some $j$, $0\le j\le t-1$,
such that $d_{C_n}(x_{i_j},x_{i_{j+1}})=f(x_{i_j})+f(x_{i_{j+1}})+1$ (subscripts are taken modulo~$t$).
Indeed, if this is not the case, we get
$$n = \sum_{0\le j\le t-1}d_{C_n}(x_{i_j},x_{i_{j+1}}) \leq 2\sum_{0\le j\le t-1}f(x_{i_j}) = 2i_b(C_n),$$
which gives $i_b(C_n)\geq \left\lceil\frac{n}{2}\right\rceil$,
in contradiction with the inequality $i_b(C_n)\leq \left\lceil\frac{2n}{5}\right\rceil$ we established before,
since $n\notin S$.

We can thus suppose, without loss of generality, that $d_{C_n}(x_{i_1},x_{i_2})=f(x_{i_1})+f(x_{i_2})+1$,
which implies that $d_{C_n}(x_{i_2},x_{i_3}) = f(x_{i_2}) +1$ (we may have $x_{i_3}=x_{i_0}$) and, similarly,
that $d_{C_n}(x_{i_0},x_{i_1}) = f(x_{i_1})+1$.
%
To avoid confusion, let us denote $P_n=y_0y_1\dots y_{n-1}$, 
with $y_{f(x_{i_2})} = x_{i_2}$,
and let $g$ be the function
defined by $g(y_j)=f(x_{j-f(x_{i_2})+i_2})$ for every $j$, $0\le j\le n-1$
(subscripts are taken modulo~$n$).
Observe that both $y_0$ and $y_{n-1}$ are $g$-dominated,
since all vertices lying between $x_{i_1}$ and $x_{i_2}$ were $f$-dominated in $C_n$.
Moreover, we cannot increase the $g$-value of $y_{f(x_{i_2})}$ (the leftmost $g$-broadcast vertex in $P_n$)
since we had $d_{C_n}(x_{i_2},x_{i_3}) = f(x_{i_2})+1$,
neither the value of $y_{n-f(x_{i_1})-1}$ (the rightmost $g$-broadcast vertex in $P_n$)
since we had $d_{C_n}(x_{i_0},x_{i_1}) = f(x_{i_1})+1$.

Since $f$ was a maximal independent broadcast on $C_n$, we thus get that $g$ is 
a maximal independent broadcast on $P_n$, which gives $i_b(P_n)\leq i_b(C_n)$,
and thus $i_b(C_n)\geq \left\lceil\frac{2n}{5}\right\rceil$, as required. 
\end{proof}

\subsection{Broadcast irredundance number and broadcast domination number}

Erwin proved in~\cite{Erw04} that  $\gamma_b(P_n) = \gamma(P_n)=\lceil n/3\rceil$. 
Knowing the value of $\gamma_b(P_n)$, we can infer the value of $\gamma_b(C_n)$. 
Indeed, Bre\v sar and \v Spacapan proved in~\cite{BS09} that, for every connected graph $G$, there is a spanning tree $T$ of $G$ 
such that $\gamma_b(G)=\gamma_b(T)$. 
Since spanning trees of the cycle $C_n$ are all isomorphic to the path $P_n$, we get the following result.

\begin{proposition}\label{prop:gamma-path-cycle}
For every integer $n\geq 3$, $\gamma_b(C_n)=\gamma_b(P_n)=\lceil\frac{n}{3}\rceil$.
\end{proposition}

We now consider the broadcast irredundance number of paths.
For that, we first prove the two following lemmas.

\begin{lemma}\label{lem:MaxIrr-H-vide}
Let $f$ be a maximal irredundant broadcast on $P_n$. If $H_f(x_{i})= \emptyset$ for some vertex $x_i$, then $N_{P_n}(x_i)\cap N_f(V^+_f)\neq \emptyset$.
\end{lemma}

\begin{proof} 
Assume to the contrary that we have $N_{P_n}(x_i)\cap N_f(V^+_f) = \emptyset$.
In that case, we could set $f(x_i)=1$, contradicting the maximality of~$f$.
\end{proof}

\begin{lemma}\label{lem:MaxIrrPn-et-irrPn}
For every integer $n\geq 3$, the following statements hold.
\begin{enumerate}
\item If $f$ is a maximal irredundant broadcast on $P_n$, then $H_f(x_2)\neq \emptyset$ and $H_f(x_{n-1})\neq \emptyset$.
\item There exists an  $ir_b$-broadcast $f$ on $P_n$ such that $H_f(x_1)\neq \emptyset$ and $H_f(x_{n})\neq \emptyset$.
\end{enumerate}
\end{lemma}

\begin{proof} 
We prove the two statements separately.
\begin{enumerate}
\item 
If $H_f(x_2)= \emptyset$, then $x_2$ is not  $f$-dominated, and consequently $x_1$ is also not  $f$-dominated. 
We then have $N_{P_n}(x_1)\cap N_f(V^+_f)= \emptyset$, in contradiction with Lemma~\ref{lem:MaxIrr-H-vide}.
The case $H_f(x_{n-1})= \emptyset$ is similar.

\item Let $g$ be an  $ir_b$-broadcast on $P_n$. 
If $H_g(x_1)\neq \emptyset$ and $H_g(x_{n})\neq \emptyset$, then we let $f:=g$ and we are done.

Suppose that we have $H_g(x_1)=\emptyset$.
By the previous item, we know that $x_2$ is $g$-dominated by some vertex $x_i$, $i>2$,
such that $f(x_i)=d_{P_n}(x_2,x_i)$. 
We then necessarily have $|PB_g(x_i)|=1$ if $g(x_i)\geq 2$, and $x_i\notin PN_g(x_i)$ if $g(x_i)=1$,
since otherwise we could set $g(x_1)=1$, contradicting the optimality of $g$.

Now, observe that the function $h$ obtained from $g$ by setting $h(x_i)=0$ and $h(x_{i-1})=g(x_i)$
is a maximal irredundant broadcast on $P_n$, 
with cost $\sigma(h)=\sigma(g)=ir_b(P_n)$, that satisfies $H_h(x_1)\neq \emptyset$.

If $H_h(x_{n})\neq \emptyset$, then we let $f:=h$ and we are done.
Otherwise, using the same reasoning (by symmetry), we claim that can produce an  $ir_b$-broadcast $f$ on $P_n$
such that $H_f(x_1)\neq \emptyset$ and $H_f(x_{n})\neq \emptyset$.
Observe first that we cannot have 
$|V_g^+|=2$ if $H_g(x_1)= \emptyset$ and $H_g(x_n)= \emptyset$, since in that case we could
increase by~1 the $g$-values of $x_{i_1}$ and $x_{i_2}$, contradicting the maximality of $g$.
Therefore, 
if $g$ was such that $H_g(x_1)=\emptyset$ and $H_g(x_n)=\emptyset$,
then the optimality of $g$ implies $|V_g^+| \geq 3$, so that the modification of $h$
does not affect $h(x_i)$, and thus $x_1$ is still $f$-dominated.
\end{enumerate}
This concludes the proof.
\end{proof}


We are now able to determine the value of the broadcast irredundance number 
and of the broadcast domination number of paths.

\begin{figure}
\begin{center}
\begin{tikzpicture}
\FLECHE{0}{0}
\LIGNE{1}{0}{8}{0}
\POINTILLE{0}{0}{1}{0}
\POINTILLE{8}{0}{9}{0}
   \dSOM{1}{0}{0}{}
   \bSOM{2}{0}{1}{}
   \dSOM{3}{0}{0}{}
   \bSOM{4}{0}{2}{{\bf 0}}
   \dSOM{5}{0}{0}{{\bf 2}}
   \dSOM{6}{0}{0}{}
   \SOM{7}{0}{0}{}
   \dSOM{8}{0}{0}{}
	
\end{tikzpicture}

(a) $x_{j'}=x_{j'^p+1}$ and $x_{i+1}$ is $f$-dominated
\vskip 0.8cm

\begin{tikzpicture}
\FLECHE{0}{0}
\LIGNE{1}{0}{8}{0}
\POINTILLE{0}{0}{1}{0}
\POINTILLE{8}{0}{9}{0}
   \dSOM{1}{0}{0}{}
   \bSOM{2}{0}{1}{}
   \dSOM{3}{0}{0}{}
   \bSOM{4}{0}{2}{{\bf 0}}
   \dSOM{5}{0}{0}{}
   \dSOM{6}{0}{0}{{\bf 2}}
   \SOM{7}{0}{0}{}
   \SOM{8}{0}{0}{}
	
\end{tikzpicture}

(b) $x_{j'}=x_{j'^p+1}$ and $x_{i+1}$ is not $f$-dominated
\vskip 0.8cm

\begin{tikzpicture}
\FLECHE{-1}{0}
\LIGNE{0}{0}{7}{0}
\POINTILLE{-1}{0}{0}{0}
\POINTILLE{7}{0}{8}{0}
   \dSOM{0}{0}{0}{}
   \dSOM{1}{0}{0}{{\bf 1}}
   \bSOM{2}{0}{2}{{\bf 0}}
   \bSOM{3}{0}{2}{{\bf 0}}
   \dSOM{4}{0}{0}{{\bf 2}}
   \dSOM{5}{0}{0}{}
   \SOM{6}{0}{0}{}
   \dSOM{7}{0}{0}{}
	
\end{tikzpicture}

(c) $x_{j'} \neq x_{j'^p+1}$ and $x_{i+1}$ is $f$-dominated
\vskip 0.8cm

\begin{tikzpicture}
\FLECHE{-1}{0}
\LIGNE{0}{0}{7}{0}
\POINTILLE{-1}{0}{0}{0}
\POINTILLE{7}{0}{8}{0}
   \dSOM{0}{0}{0}{}
   \dSOM{1}{0}{0}{{\bf 1}}
   \bSOM{2}{0}{2}{{\bf 0}}
   \bSOM{3}{0}{2}{{\bf 0}}
   \dSOM{4}{0}{0}{}
   \dSOM{5}{0}{0}{{\bf 2}}
   \SOM{6}{0}{0}{}
   \SOM{7}{0}{0}{}
	
\end{tikzpicture}

(d) $x_{j'} \neq x_{j'^p+1}$ and $x_{i+1}$ is not $f$-dominated
\caption{\label{fig:th-irPn} Maximal irredundant broadcast for the proof of Theorem~\ref{th:ir-gamma-path}.}
\end{center}
\end{figure}
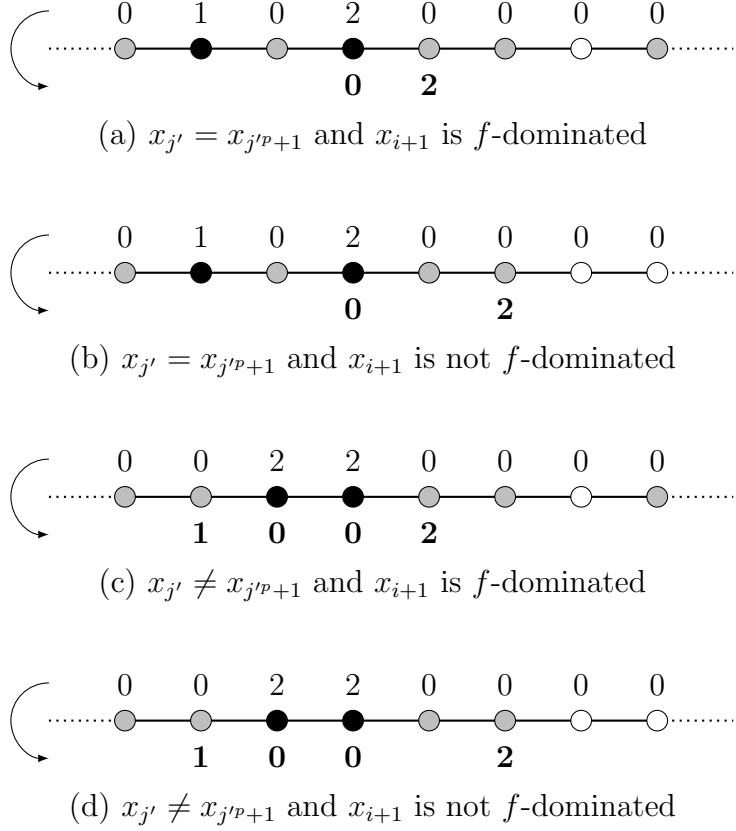

\begin{theorem}\label{th:ir-gamma-path}
For every integer $n\ge 2$, $ir_b(P_n)= \gamma_b(P_n)= \left\lceil\frac{n}{3}\right\rceil.$
\end{theorem}


\begin{proof} 
By Corollary~\ref{cor:Ahmadi-inequalities} and Proposition~\ref{prop:gamma-path-cycle}, we only need to prove that $ \gamma_b(P_n)\leq ir_b(P_n)$. 
For this, it is enough to construct, from any non-dominating   $ir_b$-broadcast, 
a dominating  $ir_b$-broadcast.

Let $f$ be an  $ir_b$-broadcast on $P_n$. By Lemma~\ref{lem:MaxIrrPn-et-irrPn},
we can assume that $x_1$, $x_2$, $x_{n-1}$ and $x_n$ are $f$-dominated.
If $f$ is dominating, then we are done.
Thus suppose that $f$ is non-dominating, and let $x_i$, $3\le i\le n-2$, be the leftmost non-dominated vertex.
We will prove that there exists a maximal irredundant broadcast $g$ on $P_n$,
with $\sigma(g)=\sigma(f)=ir_b(P_n)$,
such that 
the number of vertices that are not $g$-dominated is strictly less than
the number of vertices that are not $f$-dominated.

\medskip


Let $x_{j}$, $j\le i-2$, denote the $f$-broadcast vertex that dominates $x_{i-1}$.
Since $x_i$ is not $f$-dominated, we necessarily have $x_{i-1}\in PB(x_{j})$.
Observe that we have $|PB(x_{j})|=1$, that is, 
the bordering private $f$-neighbor of $x_j$ is $x_{i-1}$, since otherwise we could set $f(x_{i-1})=1$, 
contradicting the maximality of $f$.
Note also that $x_1$ is not $f$-dominated by $x_j$, that is, $x_j$ is not the
leftmost $f$-broadcast vertex, since otherwise we could increase the $f$-value
of $x_j$ by~1, again contradicting the maximality of $f$.
Let then $x_{j'}$, $j'<j$, denote the closest $f$-broadcast vertex to the left of $x_j$,
and $x_{j'^p}$, $j'^p<j'$, denote the bordering private $f$-neighbor of $x_{j'}$.

Since $x_{j'^p}$ is the bordering private $f$-neighbor of $x_{j'}$, we necessarily have 
$d_{P_n}(x_{j'^p},x_j) \geq f(x_j)+1$.
Moreover, we necessarily have $d_{P_n}(x_{j'^p},x_j) = f(x_j)+1$, since otherwise
we could increase the value of $f(x_j)$ by~1 ($x_i$ becoming the bordering private
$f$-neighbor of $x_j$), contradicting the maximality of $f$.
Hence, we have $d_{P_n}(x_{j'^p},x_j) = f(x_j)+1$, and thus 
$$d_{P_n}(x_{j'^p},x_i) = d_{P_n}(x_{j'^p},x_j) + d_{P_n}(x_j,x_i) = f(x_{j})+1 + f(x_{j})+1 = 2f(x_j)+2.$$

Let now $g$ be the function obtained from $f$ by setting
\begin{itemize}
\item $g(x_j)=0$, 
\item $g(x_{j'^p+1})=1$ and $g(x_{j'})=0$ if $x_{j'}\neq x_{j'^p+1}$,  and
\item $g(x_{j+1})=f(x_j)$ if $x_{i+1}$ is $f$-dominated, or $g(x_{j+2})=f(x_j)$ if $x_{i+1}$ is not $f$-dominated 
(see Figure~\ref{fig:th-irPn}).
\end{itemize}
Observe that $x_{j'^p}$ is a bordering private $g$-neighbor of $x_{j'^p+1}$,
and that either $x_i$ is a bordering private $g$-neighbor of $x_{j+1}$ (if $x_{i+1}$ is $f$-dominated),
or $x_{i+1}$ is a bordering private $g$-neighbor of $x_{j+2}$ (if $x_{i+1}$ is not $f$-dominated).
Moreover, all vertices $x_1,\dots,x_i$ are $g$-dominated.
Hence, $g$ is a maximal irredundant broadcast with cost
$$\sigma(g) = \sigma(f) - f(x_j) - f(x_{j'}) + f(x_j) + 1 = \sigma(f) - f(x_{j'}) + 1.$$
The optimality of $f$ then imply $f(x_{j'}) = 1$,
so that $g$ is an  $ir_b$-broadcast on $P_n$ such that
the number of vertices that are not $g$-dominated is strictly less than
the number of vertices that are not $f$-dominated.

By repeating this modification while they remain non-dominated vertices,
we eventually get a dominating  $ir_b$-broadcast on $P_n$, which
concludes the proof.
\end{proof}


Using Theorem~\ref{th:ir-gamma-path}, we can prove
a similar result for cycles.
%


\begin{theorem}\label{th:ir-gamma-cycle}
For every integer $n\geq 3$, $ir_b(C_n)= \gamma_b(C_n)= \left\lceil\frac{n}{3}\right\rceil$.
\end{theorem}

\begin{proof}
We already know by Proposition~\ref{prop:gamma-path-cycle} that 
$\gamma_b(C_n)= \left\lceil\frac{n}{3}\right\rceil$, so that we only need
to prove that $ir_b(C_n)=\left\lceil\frac{n}{3}\right\rceil$.
By Corollary~\ref{cor:Ahmadi-inequalities}, we only need to prove that $ ir_b(C_n)\ge \gamma_b(C_n)$ or, by 
Proposition~\ref{prop:gamma-path-cycle} and Theorem~\ref{th:ir-gamma-path},
that $ir_b(C_n) \geq ir_b(P_n)$.

Observe  that $010$, $0200$ and $00200$ are dominating
 $ir_b$-broadcasts for $C_3$, $C_4$ and $C_5$, respectively, with cost 
$\left\lceil\frac{n}{3}\right\rceil$.
It thus remains to consider the case $n\ge 6$.

Let $f$ be an  $ir_b$-broadcast on $C_n$, $n\ge 6$. 
If $f$ is dominating, then we are done.
Thus suppose that $f$ is non-dominating. 
We  claim that $|V^+_f|\ge 2$.
Indeed, if $|V^+_f|=1$, then the maximality of $f$ implies $\sigma(f)=\left\lfloor\frac{n}{2}\right\rfloor$,
which gives $ir_b(C_n) = \left\lfloor\frac{n}{2}\right\rfloor > \left\lceil\frac{n}{3}\right\rceil = \gamma_b(C_n)$,
in contradiction with Corollary~\ref{cor:Ahmadi-inequalities} since $n\ge 6$.

We thus have $|V^+_f|\ge 2$.
Since $f$ is maximal, we cannot have three consecutive vertices that are not $f$-dominated.
Moreover, since $f$ is non-dominating, we can assume, without loss of generality, that $x_0$
is not $f$-dominated, and that $x_1$ is $f$-dominated by some $f$-broadcast vertex $x_{i_1}$, $i_1>1$. 
Note that $x_{n-1}$ may be $f$-dominated or not.

To avoid confusion, let $P_n=y_0y_1\dots y_{n-1}$.
Let then $g$ be the mapping defined by $g(y_i)=f(x_i)$ for every $i$, $0\le i\le n-1$.
As in $C_n$, $x_0$ is not $g$-dominated, $x_1$ is $g$-dominated by $x_{i_1}$, 
and $g(x_{i_1})$ cannot be increased since otherwise $f(x_{i_1})$ could be increased,
contradicting the maximality of $f$.
Similarly, the $g$-value of the rightmost $g$-broadcast vertex in $P_n$
cannot be increased, since otherwise its $f$-value could be increased
(since $x_0$ is not $f$-dominated, while $x_{n-1}$ may be $f$-dominated or not), 
again contradicting the maximality of $f$.
we can apply the same reasoning if $x_{n-1}$ is not $g$-dominated,
which implies that $x_{n-1}$ is not $f$-dominated.
Hence, since $f$ is an  $ir_b$-broadcast on $C_n$,
we get that $g$ is also a maximal irredundant broadcast on $P_n$,
which gives 
$ir_b(C_n)=\sigma(f)=\sigma(g)\ge ir_b(P_n)$
as required.
\end{proof}

\subsection{Packing broadcast number and lower packing broadcast number}

We first determine the broadcast packing number 
of paths and cycles.

\begin{theorem}\label{th:P-path-cycle}
For every $n\ge 2$, $P_b(P_n) = \diam(P_n) = n-1$, and,
for every $n\ge 3$,  $P_b(C_n) = \diam(C_n) = \left\lfloor\frac{n}{2}\right\rfloor$.
\end{theorem}

\begin{proof} 
%
%
%
We first consider the case of the path $P_n$, $n\ge 2$.
Observe first that the function $f$ defined by $f(x_1)=n-1$ and $f(x_i)=0$ for every $i$, $2\le i\le n$
is a maximal broadcast packing with cost $n-1=\diam(P_n)$, which gives $\diam(P_n)\leq P_b(P_n)$.
The opposite inequality directly follows from Observation~\ref{obs:definitions}(3) and Theorem~\ref{th:Gamma_IR_paths}.

%

\medskip

Let us now consider the case of the cycle $C_n$, $n\geq 3$.
Observe first that the function $f$ defined by $f(x_0)=\left\lfloor\frac{n}{2}\right\rfloor$ and $f(x_i)=0$ for every $i$, $1\le i\le n-1$
is a maximal broadcast packing with cost $\left\lfloor\frac{n}{2}\right\rfloor=\diam(C_n)$, 
which gives $\diam(C_n)\leq P_b(C_n)$.

Again, to establish the opposite inequality, it suffices to prove that, for every  $P_b$-broadcast $f$ on $C_n$, $|V_f^+|=1$.
Similarly as above, if we suppose that $f$ is a  $P_b$-broadcast on $C_n$ with $|V_f^+|\geq 2$,
we get
$$\sum_{x_i\in V_f^+}\big(2f(x_i)+1\big)\leq n,$$
which gives
$$ 2P_b(C_n)+|V_f^+| \leq n,$$
and thus
$$P_b(C_n)\leq \frac{n-|V_f^+|}{2}\leq \frac{n-2}{2}< \left\lfloor\frac{n}{2}\right\rfloor= {\mbox{ diam}}(C_n),$$
again a contradiction.
\end{proof}


In order to determine the values of $p_b(P_n)$, $n\ge 1$, we first prove the following lemma.

\begin{lemma}\label{lem:packing 0-1}
For every integer $n\ge 2$,  there exists a  $p_b$-broadcast $f$ on $P_n$ such that  $f(x_i)=1$ 
for every $f$-broadcast vertex $x_i$.
\end{lemma}

\begin{proof} 
Observe first that $10$, $010$, $1001$ and $10010$ define
$p_b$-broadcasts on $P_n$, $n=2,3,4,5$, respectively, that satisfy the statement of the lemma.
Suppose thus $n \geq 6$ and let $g$ be any   $p_b$-broadcast on $P_n$. 
If $g(x_i)=1$ for every $g$-broadcast vertex $x_i$, then we set $f:=g$ and we are done.

Otherwise, let $V_g^+=\{x_{i_1},\dots,x_{i_t}\}$, $i_1<\cdots < i_t$, $t\geq 1$.
We first claim that we necessarily have $t\ge 2$.
Indeed, if $t=1$, we get 
$$p_b(P_n)=g(x_{i_1})=\rad(P_n)=\left\lfloor\frac{n}{2}\right\rfloor,$$
while the function 
$g'$ defined by $g'(x_i)=1$, $1\le i\le n$, if and only if $i \equiv 1\pmod 3$
is a maximal packing broadcast with cost
$\sigma(g')=\left\lfloor\frac{n}{3}\right\rfloor < \left\lfloor\frac{n}{2}\right\rfloor,$
a contradiction.

We thus have $|V_g^+| \geq 2$. Let $x_{i_j}$, $1\le j\le t$, be a $g$-broadcast vertex with minimum subscript such that
$g(x_{i_j})\geq 2$. We will consider three cases, depending on the value of $i_j$.
In each case, we will prove either that the case cannot occur,
or that we can produce a $p_b$-broadcast $g'$ on $P_n$,
with $\sigma(g')=\sigma(g)$, 
such that the subscript of the leftmost $g'$-broadcast vertex with $g'$-value at least~$2$, if any, is
strictly greater than
the subscript of the leftmost $g$-broadcast vertex with $g$-value at least~$2$.


\begin{enumerate}
\item $i_j=i_1$ or $i_j=i_t$.\\ 
Assume $i_j=i_1$, the case $i_j=i_t$ being similar, by symmetry.
We first claim that we have $g(x_{i_1})\in\{i_1-2, i_1-1\}$.
Indeed, if $g(x_{i_1}) \geq i_1$, then the function
$h$ obtained from $g$ by setting $h(x_{i_1})=0$ and $h(x_{i_1+1})=g(x_{i_1})-1$
is clearly a maximal packing broadcast with cost $\sigma(h)=\sigma(g)-1$,
contradicting the optimality of $g$.
Now, if $g(x_{i_1}) \leq i_1 - 3$, then 
we could set $g(x_{1})=1$, contradicting the maximality of $g$.
We thus have two cases to consider.

\begin{figure}
\begin{center}
\begin{tikzpicture}
\FLECHE{-0.5}{0}
\LIGNE{0}{0}{4}{0}
\POINTILLE{4}{0}{5}{0}
   \dSOM{0}{0}{0}{{\bf 1}}
   \dSOM{1}{0}{0}{{\bf 0}}
   \bSOM{2}{0}{2}{{\bf 0}}
   \dSOM{3}{0}{0}{{\bf 1}}
   \dSOM{4}{0}{0}{{\bf 0}}
\end{tikzpicture}

\vskip 0.4cm

\begin{tikzpicture}
\FLECHE{-0.5}{0}
\LIGNE{0}{0}{6}{0}
\POINTILLE{6}{0}{7}{0}
   \dSOM{0}{0}{0}{{\bf 0}}
   \dSOM{1}{0}{0}{{\bf 0}}
   \dSOM{2}{0}{0}{{\bf 1}}
   \bSOM{3}{0}{3}{{\bf 0}}
   \dSOM{4}{0}{0}{{\bf 0}}
   \dSOM{5}{0}{0}{{\bf 1}}
   \dSOM{6}{0}{0}{{\bf 0}}
\end{tikzpicture}

\vskip 0.4cm

\begin{tikzpicture}
\FLECHE{-0.5}{0}
\LIGNE{0}{0}{8}{0}
\POINTILLE{8}{0}{9}{0}
   \dSOM{0}{0}{0}{{\bf 0}}
   \dSOM{1}{0}{0}{{\bf 1}}
   \dSOM{2}{0}{0}{{\bf 0}}
   \dSOM{3}{0}{0}{{\bf 0}}
   \bSOM{4}{0}{4}{{\bf 1}}
   \dSOM{5}{0}{0}{{\bf 0}}
   \dSOM{6}{0}{0}{{\bf 0}}
   \dSOM{7}{0}{0}{{\bf 1}}
   \dSOM{8}{0}{0}{{\bf 0}}
\end{tikzpicture}

(a) $g(x_{i_1})=i_1-1$

\vskip 0.8cm

\begin{tikzpicture}
\FLECHE{-1.5}{0}
\LIGNE{-1}{0}{4}{0}
\POINTILLE{4}{0}{5}{0}
   \SOM{-1}{0}{0}{{\bf 0}}
   \dSOM{0}{0}{0}{{\bf 1}}
   \dSOM{1}{0}{0}{{\bf 0}}
   \bSOM{2}{0}{2}{{\bf 0}}
   \dSOM{3}{0}{0}{{\bf 1}}
   \dSOM{4}{0}{0}{{\bf 0}}
\end{tikzpicture}

\vskip 0.4cm

\begin{tikzpicture}
\FLECHE{-1.5}{0}
\LIGNE{-1}{0}{6}{0}
\POINTILLE{6}{0}{7}{0}
   \SOM{-1}{0}{0}{{\bf 1}}
   \dSOM{0}{0}{0}{{\bf 0}}
   \dSOM{1}{0}{0}{{\bf 0}}
   \dSOM{2}{0}{0}{{\bf 1}}
   \bSOM{3}{0}{3}{{\bf 0}}
   \dSOM{4}{0}{0}{{\bf 0}}
   \dSOM{5}{0}{0}{{\bf 1}}
   \dSOM{6}{0}{0}{{\bf 0}}
\end{tikzpicture}

\vskip 0.4cm

\begin{tikzpicture}
\FLECHE{-1.5}{0}
\LIGNE{-1}{0}{8}{0}
\POINTILLE{8}{0}{9}{0}
   \SOM{-1}{0}{0}{{\bf 0}}
   \dSOM{0}{0}{0}{{\bf 0}}
   \dSOM{1}{0}{0}{{\bf 1}}
   \dSOM{2}{0}{0}{{\bf 0}}
   \dSOM{3}{0}{0}{{\bf 0}}
   \bSOM{4}{0}{4}{{\bf 1}}
   \dSOM{5}{0}{0}{{\bf 0}}
   \dSOM{6}{0}{0}{{\bf 0}}
   \dSOM{7}{0}{0}{{\bf 1}}
   \dSOM{8}{0}{0}{{\bf 0}}
\end{tikzpicture}

(b)  $g(x_{i_1})=i_1-2$

\caption{\label{fig:lem-packing}Packing broadcasts for the proof of Lemma~\ref{lem:packing 0-1}, Case 1.}
\end{center}
\end{figure}

\begin{enumerate}
\item $g(x_{i_1})= i_1-1$, and thus $i_1\geq 3$. \\ 
Let $g_1$ be the function obtained from $g$ by modifying the $g$-values of
$x_1,\dots, x_{2g(x_{i_1})+1}$ as follows:
\begin{itemize}
\item $g_1(x_1\dots x_{2g(x_{i_1})+1}) = (010)^\alpha$, if  $2g(x_{i_1})+1\equiv 0 \pmod 3$,
\item $g_1(x_1\dots x_{2g(x_{i_1})+1}) = 0(010)^\alpha$, if  $2g(x_{i_1})+1\equiv 1 \pmod 3$,
\item $g_1(x_1\dots x_{2g(x_{i_1})+1}) = 10(010)^\alpha$, if  $2g(x_{i_1})+1\equiv 2 \pmod 3$,
\end{itemize}
where $\alpha=\left\lfloor\frac{2g(x_{i_1})+1}{3}\right\rfloor$ (see Figure~\ref{fig:lem-packing}(a)).
It is then not difficult to check that $g_1$ is a maximal packing broadcast such that
$$\sigma(g_1)-\sigma(g)= \Big(\left\lfloor\frac{2g(x_{i_1})-1}{3}\right\rfloor +1\Big) - g(x_{i_1}) = \left\lfloor\frac{2-g(x_{i_1})}{3}\right\rfloor,$$
which gives 
$\sigma(g_1)-\sigma(g)<0$ whenever $g(x_{i_1})\neq 2$. 
Consequently, $g(x_{i_1})=g(x_{3})=2$ and we can then define $g'$ 
from $g$ by setting $g'(x_1\dots x_5)= 10010$.


\item $g(x_{i_1})= i_1-2$, and thus $i_1\geq 4$. \\ 
Similarly to the previous case, let $g_2$ be the function obtained from $g$ by modifying the $g$-values of
$x_1,\dots, x_{2g(x_{i_1})+2}$ as follows:
\begin{itemize}
\item $g_2(x_1\dots x_{2g(x_{i_1})+2}) = (010)^\alpha$, if  $2g(x_{i_1})+2\equiv 0 \pmod 3$,
\item $g_2(x_1\dots x_{2g(x_{i_1})+2}) = 0(010)^\alpha$, if  $2g(x_{i_1})+2\equiv 1 \pmod 3$,
\item $g_2(x_1\dots x_{2g(x_{i_1})+2}) = 10(010)^\alpha$, if  $2g(x_{i_1})+2\equiv 2 \pmod 3$,
\end{itemize}
where $\alpha=\left\lfloor\frac{2g(x_{i_1})+2}{3}\right\rfloor$ (see Figure~\ref{fig:lem-packing}(b))
.
Again, it is not difficult to check that $g_2$ is a maximal packing broadcast such that
$$\sigma(g_2)-\sigma(g)= \Big(\left\lfloor\frac{2g(x_{i_1})}{3}\right\rfloor +1\Big) - g(x_{i_1}) = \left\lfloor\frac{3-g(x_{i_1})}{3}\right\rfloor,$$
which gives 
$\sigma(g_2)-\sigma(g)<0$ whenever $g(x_{i_1})\neq 2$ and $g(x_{i_1})\neq 3$. 
Consequently, either $g(x_{i_1})=g(x_{4})=2$ and we can then define $g'$ 
from $g$ by setting $g'(x_1\dots x_6)= 010010$,
or $g(x_{i_1})=g(x_{5})=3$ and we can then define $g'$ 
from $g$ by setting $g'(x_1\dots x_8)= 10010010$.
\end{enumerate}

\item $i_j\in \{i_2,\dots,i_{t-1}\}$. \\ 
%
%
\begin{figure}
\begin{center}
\begin{tikzpicture}
\FLECHE{-1}{0}
\LIGNE{0}{0}{8}{0}
\POINTILLE{-1}{0}{0}{0}
\POINTILLE{8}{0}{9}{0}
   \dSOM{0}{0}{0}{{\bf 0}}
   \dSOM{1}{0}{0}{{\bf 1}}
   \dSOM{2}{0}{0}{{\bf 0}}
   \dSOM{3}{0}{0}{{\bf 0}}
   \bSOM{4}{0}{4}{{\bf 1}}
   \dSOM{5}{0}{0}{{\bf 0}}
   \dSOM{6}{0}{0}{{\bf 0}}
   \dSOM{7}{0}{0}{{\bf 1}}
   \dSOM{8}{0}{0}{{\bf 0}}
\end{tikzpicture}

(a) $2g(x_{i_j})+1\equiv 0\pmod 3$ ($9$ in this example)

\vskip 0.8cm

\begin{tikzpicture}
\FLECHE{-3}{0}
\LIGNE{-2}{0}{8}{0}
\POINTILLE{-3}{0}{-2}{0}
\POINTILLE{8}{0}{9}{0}
   \bSOM{-2}{0}{1}{}
   \dSOM{-1}{0}{0}{}
   \SOM{0}{0}{0}{{\bf 0}}
   \SOM{1}{0}{0}{{\bf 1}}
   \dSOM{2}{0}{0}{{\bf 0}}
   \dSOM{3}{0}{0}{{\bf 0}}
   \dSOM{4}{0}{0}{{\bf 1}}
   \bSOM{5}{0}{3}{{\bf 0}}
   \dSOM{6}{0}{0}{{\bf 0}}
   \dSOM{7}{0}{0}{{\bf 1}}
   \dSOM{8}{0}{0}{{\bf 0}}

\FLECHE{-1}{-2}
\LIGNE{0}{-2}{7}{-2}
\POINTILLE{-1}{-2}{0}{-2}
\POINTILLE{7}{-2}{8}{-2}
   \bSOM{0}{-2}{1}{}
   \dSOM{1}{-2}{0}{}
   \SOM{2}{-2}{0}{{\bf 0}}
   \dSOM{3}{-2}{0}{{\bf 1}}
   \dSOM{4}{-2}{0}{{\bf 0}}
   \bSOM{5}{-2}{2}{{\bf 0}}
   \dSOM{6}{-2}{0}{{\bf 1}}
   \dSOM{7}{-2}{0}{{\bf 0}}

\FLECHE{-1-0.5}{-4}
\LIGNE{0-0.5}{-4}{8-0.5}{-4}
\POINTILLE{-1-0.5}{-4}{0-0.5}{-4}
\POINTILLE{8-0.5}{-4}{9-0.5}{-4}
   \bSOM{0-0.5}{-4}{1}{}
   \dSOM{1-0.5}{-4}{0}{}
   \SOM{2-0.5}{-4}{0}{}
   \SOM{3-0.5}{-4}{0}{{\bf 0}}
   \dSOM{4-0.5}{-4}{0}{{\bf 1}}
   \dSOM{5-0.5}{-4}{0}{{\bf 0}}
   \bSOM{6-0.5}{-4}{2}{{\bf 0}}
   \dSOM{7-0.5}{-4}{0}{{\bf 1}}
   \dSOM{8-0.5}{-4}{0}{{\bf 0}}

\end{tikzpicture}

(b) $2g(x_{i_j})+1\equiv 1\pmod 3$ ($7$ in this example) and $p=2$, 

or $2g(x_{i_j})+1\equiv 2\pmod 3$ ($5$ in this example) and $p=1$,

or $2g(x_{i_j})+1\equiv 2\pmod 3$ ($5$ in this example) and $p=2$

\vskip 0.8cm

\begin{tikzpicture}
\FLECHE{0}{0}
\LIGNE{1}{0}{15}{0}
\POINTILLE{0}{0}{1}{0}
\POINTILLE{15}{0}{16}{0}
   \bSOM{1}{0}{1}{}
   \dSOM{2}{0}{0}{}
   \dSOM{3}{0}{0}{{\bf 0}}
   \dSOM{4}{0}{0}{{\bf 1}}
   \dSOM{5}{0}{0}{{\bf 0}}
   \bSOM{6}{0}{3}{{\bf 0}}
   \dSOM{7}{0}{0}{{\bf 1}}
   \dSOM{8}{0}{0}{{\bf 0}}
   \dSOM{9}{0}{0}{{\bf 0}}
   \dSOM{10}{0}{0}{{\bf 0}}
   \bSOM{11}{0}{1}{{\bf 2}}
   \dSOM{12}{0}{0}{}
   \SOM{13}{0}{0}{}
   \dSOM{14}{0}{0}{}
   \bSOM{15}{0}{1}{}
\end{tikzpicture}

\vskip 0.4cm

\begin{tikzpicture}
\FLECHE{0}{0}
\LIGNE{1}{0}{16}{0}
\POINTILLE{0}{0}{1}{0}
\POINTILLE{16}{0}{17}{0}
   \bSOM{1}{0}{1}{}
   \dSOM{2}{0}{0}{}
   \dSOM{3}{0}{0}{{\bf 0}}
   \dSOM{4}{0}{0}{{\bf 1}}
   \dSOM{5}{0}{0}{{\bf 0}}
   \bSOM{6}{0}{3}{{\bf 0}}
   \dSOM{7}{0}{0}{{\bf 1}}
   \dSOM{8}{0}{0}{{\bf 0}}
   \dSOM{9}{0}{0}{{\bf 0}}
   \dSOM{10}{0}{0}{}
   \dSOM{11}{0}{0}{}
   \bSOM{12}{0}{2}{}
   \dSOM{13}{0}{0}{}
   \dSOM{14}{0}{0}{}
   \dSOM{15}{0}{0}{}
   \bSOM{16}{0}{1}{}
\end{tikzpicture}

\vskip 0.4cm

\begin{tikzpicture}
\FLECHE{0}{0}
\LIGNE{1}{0}{12}{0}
\POINTILLE{0}{0}{1}{0}
\POINTILLE{12}{0}{13}{0}
   \bSOM{1}{0}{1}{}
   \dSOM{2}{0}{0}{}
   \dSOM{3}{0}{0}{{\bf 0}}
   \dSOM{4}{0}{0}{{\bf 1}}
   \dSOM{5}{0}{0}{{\bf 0}}
   \bSOM{6}{0}{3}{{\bf 0}}
   \dSOM{7}{0}{0}{{\bf 1}}
   \dSOM{8}{0}{0}{{\bf 0}}
   \dSOM{9}{0}{0}{{\bf 0}}
   \SOM{10}{0}{0}{{\bf 0}}
   \dSOM{11}{0}{0}{{\bf 2}}
   \bSOM{12}{0}{1}{{\bf 0}}
\end{tikzpicture}

(c) $2g(x_{i_j})+1\equiv 1\pmod 3$ ($7$ in this example) and $p=0$,

or $2g(x_{i_j})+1\equiv 1\pmod 3$ ($7$ in this example), $p=0$,\\ and both $x_{i_j+g(x_{i_j})+1}$
and $x_{i_{j+1}+g(x_{i_{j+1}})+1}$ are $g$-dominated,

or $2g(x_{i_j})+1\equiv 1\pmod 3$ and $p=1$

\vskip 0.8cm

\begin{tikzpicture}
\FLECHE{0}{0}
\LIGNE{1}{0}{10}{0}
\POINTILLE{0}{0}{1}{0}
\POINTILLE{10}{0}{11}{0}
   \bSOM{1}{0}{1}{}
   \dSOM{2}{0}{0}{}
   \dSOM{3}{0}{0}{{\bf 0}}
   \dSOM{4}{0}{0}{{\bf 1}}
   \bSOM{5}{0}{2}{{\bf 0}}
   \dSOM{6}{0}{0}{{\bf 0}}
   \dSOM{7}{0}{0}{{\bf 0}}
   \dSOM{8}{0}{0}{{\bf 0}}
   \dSOM{9}{0}{0}{{\bf 3}}
   \bSOM{10}{0}{2}{{\bf 0}}
\end{tikzpicture}

(d) $2g(x_{i_j})+1\equiv 2\pmod 3$ ($5$ in this example) and $p=0$

\caption{\label{fig:lem-packing-2}Packing broadcasts for the proof of Lemma~\ref{lem:packing 0-1}, Case 2.}
\end{center}
\end{figure}

Since $g$ is a maximal packing broadcast, 
we necessarily have 
$$1\le d_{P_n}(x_{i_{j-1}},x_{i_{j}}) - g(x_{i_{j-1}}) - g(x_{i_{j}}) \le 3.$$
Moreover, we also have either 
$$d_{P_n}(x_{i_{j-1}},x_{i_{j}})=g(x_{i_{j-1}})+g(x_{i_{j}})+1,\ \  
\text{or}\ \ d_{P_n}(x_{i_{j}},x_{i_{j+1}})=g(x_{i_{j}})+g(x_{i_{j+1}})+1.$$ 


We consider four subcases, depending on the value of $2g(x_{i_j})+1\mod 3$, and on
the number $p$ of vertices lying between $x_{i_{j-1}}$ and $x_{i_{j}}$,
or between $x_{i_{j}}$ and $x_{i_{j+1}}$,
that are not $g$-dominated.
Note that since $g$ is a maximal packing broadcast, we have either 
(i) $p=0$,
or (ii) $p=1$ and either $x_{i_j-g(x_{i_j})-1}$ or $x_{i_j+g(x_{i_j})+1}$ is not $g$-dominated,
or (iii) $p=2$ and either $x_{i_j-g(x_{i_j})-2}$ and $x_{i_j-g(x_{i_j})-1}$,
or $x_{i_j+g(x_{i_j})+1}$ and $x_{i_j+g(x_{i_j})+2}$, are not $g$-dominated.

\begin{enumerate}
\item $2g(x_{i_j})+1\equiv 0\pmod 3$.\\
Let $g_0$ be the function obtained from $g$ by 
setting $g_0(x_{i_j - g(x_{i_j})} \dots x_{i_j + g(x_{i_j})}) = (010)^\alpha$, 
where $\alpha =\frac{2g(x_{i_j})+1}{3}$ (see Figure~\ref{fig:lem-packing-2}(a)).
Since $g$ is maximal, $g_0$ is also maximal and we have
$$\sigma(g_0)-\sigma(g)= \frac{2g(x_{i_j})+1}{3}- g(x_{i_j}) =\frac{1-g(x_{i_j})}{3} < 0,$$
which contradicts the optimality of $g$.

\item $2g(x_{i_j})+1\equiv 1\pmod 3$ and $p=2$, 
or $2g(x_{i_j})+1\equiv 2\pmod 3$ and $1\le p\le 2$.\\ 
Suppose that $2g(x_{i_j})+1\equiv 1\pmod 3$, and $x_{i_j-g(x_{i_j})-2}$, $x_{i_j-g(x_{i_j})-1}$ are not $g$-dominated,
or that $2g(x_{i_j})+1\equiv 2\pmod 3$, and $x_{i_j-g(x_{i_j})-1}$ is not $g$-dominated.
(The other cases are similar.)

Let $g'$ be the function obtained from $g$ by 
setting
\begin{itemize}
\item $g'(x_{i_j-g(x_{i_j})-2} \dots x_{i_j + g(x_{i_j})}) = (010)^\alpha$ if $2g(x_{i_j})+1\equiv 1\pmod 3$,
\item $g'(x_{i_j-g(x_{i_j})-1} \dots x_{i_j + g(x_{i_j})}) = (010)^\alpha$ if $2g(x_{i_j})+1\equiv 2\pmod 3$,
\end{itemize}
where $\alpha = \left\lfloor\frac{2g(x_{i_j})+p+1}{3}\right\rfloor$ (see Figure~\ref{fig:lem-packing-2}(b)).
Since $g$ is maximal, $g'$ is also maximal and we have
$$\sigma(g')-\sigma(g)= \frac{2g(x_{i_j})+p+1}{3}- g(x_{i_j})=\frac{p+1-g(x_{i_j})}{3}\leq 0.$$

The optimality of $g$ then implies either $p=1$ and $g(x_{i_j})=2$, or $p=2$ and $g(x_{i_j})=3$.
In each case, $g'$ is also optimal and the subscript of the leftmost $g'$-broadcast vertex with 
$g'$-value at least~$2$, if any, is strictly greater than
the subscript of the leftmost $g$-broadcast vertex with $g$-value at least~$2$, as required.

\item $2g(x_{i_j})+1\equiv 1\pmod 3$ and $0\le p\le 1$.\\
In that case, $x_{i_j-g(x_{i_j})-2}$ and $x_{i_j+g(x_{i_j})+2}$ are $g$-dominated,
and at most one vertex among $x_{i_j-g(x_{i_j})-1}$ and $x_{i_j+g(x_{i_j})+1}$ is not $g$-dominated.
Let $g'$ be the function obtained from $g$ (see Figure~\ref{fig:lem-packing-2}(c)) by setting
$g'(x_{i_j-g(x_{i_j})} \dots x_{i_j + g(x_{i_j})}) = (010)^\alpha 0$,  where $\alpha = \frac{2g(x_{i_j})}{3}$,
and 
\begin{itemize}
\item
$g'(x_{i_{j+1}})=0$, $g'(x_{i_{j+1}-1})= g(x_{i_{j+1}})+1$, if $x_{i_{j} + g(x_{i_{j}})+1}$ is not $g$-dominated,
\item $g'(x_{i_{j+1}})=g(x_{i_{j+1}})+1$,  if $i_j=i_{t-1}$ or $x_{i_{j} + g(x_{i_{j}})+1}$ is $g$-dominated, 
and $x_{i_{j+1} + g(x_{i_{j+1}})+1}$ is not $g$-dominated,
\end{itemize}
Note that we necessarily have one of the above cases, since otherwise $g'$ would be a maximal
packing broadcast with $\sigma(g')<\sigma(g)$, contradicting the optimality of $g$.
Since $g$ is maximal, $g'$ is also maximal and we have
$$\sigma(g')-\sigma(g)= \frac{2g(x_{i_j})}{3}+1- g(x_{i_j})=\frac{3-g(x_{i_j})}{3}\leq 0.$$
(Recall that since $2g(x_{i_j})+1\equiv 1\pmod 3$, we have $g(x_{i_j})\ge 3$.)
The optimality of $g$ then implies $g(x_{i_j})=3$.
We then get that $g'$ is also optimal and the subscript of the leftmost $g'$-broadcast vertex with 
$g'$-value at least~$2$, if any, is strictly greater than
the subscript of the leftmost $g$-broadcast vertex with $g$-value at least~$2$, as required.

\item $2g(x_{i_j})+1\equiv 2\pmod 3$ and $p=0$.\\ 
Let $g'$ be the function obtained from $g$ by setting
$g'(x_{i_j-g(x_{i_j})} \dots x_{i_j + g(x_{i_j})}) = (010)^\alpha 00$, where $\alpha = \frac{2g(x_{i_j})-1}{3}$,
$g'(x_{i_{j+1}})=0$, and $g'(x_{i_{j+1}-1})=g(x_{i_{j+1}})+1$ (see Figure~\ref{fig:lem-packing-2}(d)). 
Since $g$ is maximal, $g'$ is also maximal and we have
$$\sigma(g')-\sigma(g)= \frac{2g(x_{i_j})-1}{3}+1- g(x_{i_j})=\frac{2-g(x_{i_j})}{3}\leq 0.$$
The optimality of $g$ then implies $g(x_{i_j})=2$.
Therefore, $g'$ is also optimal and the subscript of the leftmost $g'$-broadcast vertex with 
$g'$-value at least~$2$, if any, is strictly greater than
the subscript of the leftmost $g$-broadcast vertex with $g$-value at least~$2$, as required.
\end{enumerate}
%
\end{enumerate}

Repeating the same transformation for each vertex with $g$-value at least~2, 
we eventually produce a $p_b$-broadcast $g'$ on $P_n$ all of whose
broadcast vertices have $g'$-value~1, as claimed in the statement of the lemma.
This concludes the proof.
\end{proof}

We are now able to determine the lower broadcast packing number of paths.

\begin{theorem}\label{th:p-path}
For every integer $n\ge 2$,
\[p_b(P_n)=\left\lbrace\begin{array}{cl}
 \frac{n}{4}&  \text{if }  \,\, n\equiv 0\pmod 8,\\[1ex]
 2\left\lfloor\frac{n}{8}\right\rfloor +1&  \text{if }  \,\, n\equiv 1,2,3\pmod 8,\\[1ex]
 2\left\lfloor\frac{n}{8}\right\rfloor +2&  \text{if }  \,\, n\equiv 4,5,6,7\pmod 8.\\
                 \end{array}
\right.
\]
\end{theorem}

\begin{proof} 
%
Observe first that $10$, $010$, $1001$, $01001$, $001001$, $0010010$ and $00100100$ define
optimal $p_b$-broadcasts on $P_n$ for $n=2,\dots,8$, respectively, whose costs are the values claimed
by the theorem.

Suppose now $n\ge 9$ and let
$n=8q+r$, with $q\geq 1$ and $0\leq r\leq 7$.
According to the value of $r$,  we define the broadcasts $f_r$, $0\le r\le 7$, on $P_n$ as follows:
$$f_0(P_n) = (00100100)^q,\ f_1(P_n) = (00100100)^q1,\ f_2(P_n) = (00100100)^q10,$$
$$f_3(P_n) = (00100100)^q100,\ f_4(P_n) = (00100100)^q1001,\ f_5(P_n) = (00100100)^q01001,$$
$$f_6(P_n) = (00100100)^q001001,\ \mbox{and } f_7(P_n) = (00100100)^q0010010.$$

It is not difficult to check that each $f_r$ is indeed a maximal packing broadcast 
on $P_{8q+r}$, $8q+r\ge 9$, with cost
$$\sigma(f_r)=
\left\{
\begin{array}{cl}
 \frac{n}{4}&  if  \,\, n\equiv 0\pmod 8,\\[1ex]
 2\left\lfloor\frac{n}{8}\right\rfloor +1&  \text{if }  \,\, n\equiv 1,2,3\pmod 8,\\[1ex]
 2\left\lfloor\frac{n}{8}\right\rfloor +2&  \text{if }  \,\, n\equiv 4,5,6,7\pmod 8,
 \end{array}
\right.$$
which gives $p_b(P_n)\leq \sigma(f_r)$, that is, $p_b(P_n)$ is not greater than the value claimed
by the theorem.

\medskip

We now prove the opposite inequality.
By Lemma~\ref{lem:packing 0-1}, we know that there exists a   $p_b$-broadcast all of whose
broadcast vertices have broadcast value~1. Let $g$ be such a broadcast.
For every $k$, $0\leq k\leq q-1$, let
$$\sigma_k = g(x_{8k+1})+\cdots +g(x_{8k+8}).$$

From the definition of a packing broadcast, we get that the distance between any two
consecutive $g$-broadcast vertices (with $g$-value $1$) is $3$, $4$ or $5$, which implies
$2\leq \sigma_k \leq 3$ for every $k$, $0\leq k\leq q-1$.
%
%
Observe also that $x_2$ and $x_{n-1}$ must be $g$-dominated, since otherwise we could
set $g(x_1)=1$ or $g(x_n)=1$, contradicting the maximality of $g$,
and that 
$d_{P_n}(x_{i_1},x_{i_2})=d_{P_n}(x_{i_{t-1}},x_{i_t})=3$, 
since otherwise we could increase the value of $g(x_{i_1})$ or $g(x_{i_t})$, 
contradicting the maximality of $g$.

We now consider three cases, depending on the value of $r$, and prove in each case that
the cost of~$g$ is the value claimed in the statement of the theorem.

\begin{enumerate}
\item $r=0$.\\
In that case, we have
$$p_b(P_n)=\sigma(g)=\sum_{k=0}^{q-1}\sigma_k\geq 2q = \frac{2n}{8}= \frac{n}{4}.$$

\item $r\in\{1,2,3\}$.\\
If $\sigma_k=3$ for some $k$, $0\le k\le q-1$, then
$$p_b(P_n)=\sigma(g) \geq \sum_{k=0}^{q-1}\sigma_k \geq 2(q-1)+3=2q+1= 2\left\lfloor\frac{n}{8}\right\rfloor +1,$$
and we are done.

%

Suppose now that $\sigma_k=2$ for every $k$, $0\leq k\leq q-1$, 
which implies $p_b(P_n)\ge 2\left\lfloor\frac{n}{8}\right\rfloor$.
Suppose, contrary to the statement of the theorem, that $p_b(P_n) = 2\left\lfloor\frac{n}{8}\right\rfloor$.
This implies  $g(x_{n-r+1})=\cdots=g(x_n)=0$.
If $r=3$, then we have a contradiction since $x_{n-1}$ must be $g$-dominated.

We thus have $r\in\{1,2\}$. Since $x_{n-1}$ must be $g$-dominated and
$d_{P_n}(x_{i_{t-1}},x_{i_t})=3$, we necessarily have
$\sigma_{q-1}\in\{00001001,00010010\}$ if $r=1$,
and $\sigma_{q-1}=00001001$ if $r=2$.
Since $g$ is maximal, we cannot have three consecutive vertices that
are not $g$-dominated, which gives
$\sigma_{j}\in\{00001001,00010010\}$ if $r=1$,
and $\sigma_{j}=00001001$ if $r=2$, for every $j$, $0\le j\le q-1$,
a contradiction since, in each case, this would imply that $x_2$ is not $g$-dominated.

Hence we have $p_b(P_n) = 2\left\lfloor\frac{n}{8}\right\rfloor +1$, as required.

%
%

\item $r \in \{4,5,6,7\}$.\\
Note first that we necessarily have $g(x_{n-r+1})+\cdots +g(x_n) \ge 1$.
Hence, if $\sigma_k=3$ for some $k$, $0\le k\le q-1$, then
$$p_b(P_n)=\sigma(g) \geq \sum_{k=0}^{q-1}\sigma_k + 1 \geq 2(q-1)+3+1=2q+2= 2\left\lfloor\frac{n}{8}\right\rfloor +2,$$
and we are done.

Suppose now that $\sigma_k=2$ for every $k$, $0\leq k\leq q-1$,
and, contrary to the statement of the theorem, 
that $p_b(P_n) \le 2\left\lfloor\frac{n}{8}\right\rfloor + 1$,
which implies $g(x_{n-r+1})+\cdots+g(x_n)=1$ since $x_{n-1}$ must be $g$-dominated.
If $r=6$ or $r=7$, then we have a contradiction since we must have $d_{P_n}(x_{i_{t-1}},x_{i_t})=3$.

We thus have $r\in\{4,5\}$. 
Again, since $x_{n-1}$ must be $g$-dominated and
$d_{P_n}(x_{i_{t-1}},x_{i_t})=3$, we necessarily have
$$\sigma_{q-1}\in\{00100001,00001001,01000010,00100010,00010010\}$$ 
if $r=4$, and 
$$\sigma_{q-1}\in\{00100001,00001001\}$$ 
if $r=5$.
Now, since $g$ is maximal, every $g$-broadcast vertex must
be at distance~$3$ from another $g$-broadcast vertex, and thus
we get 
$$\sigma_{j}\in\{00100001,00001001,01000010,00100010,00010010\}$$ 
for every $j$, $0\le j\le q-1$.
We then get a contradiction since $x_2$ must be $g$-dominated and we must have $d(x_{i_1},x_{i_2})=3$.

Hence, in this case also, we have $p_b(P_n) = 2\left\lfloor\frac{n}{8}\right\rfloor +2$, as required.
\end{enumerate}

This completes the proof.
\end{proof}

%
%

Using Theorem~\ref{th:p-path}, we can also prove a similar result for cycles.

\begin{theorem}\label{th:p-cycle}
For every integer $n\geq 3$,
\[
p_b(C_n)=\left\lbrace\begin{array}{cl}
 \frac{n}{4}&  \text{if }  \,\, n\equiv 0\pmod 8,\\[1ex]
 2\lfloor\frac{n}{8}\rfloor +1&  \text{if }  \,\, n\equiv 1,2,3\pmod 8,\\[1ex]
 2\lfloor\frac{n}{8}\rfloor +2&  \text{if }  \,\, n\equiv 4,5,6,7\pmod 8.
                 \end{array}
\right.
\]
\end{theorem}

\begin{proof}
Observe first that the broadcasts $f_0,\dots,f_7$, defined in the proof of Theorem~\ref{th:p-path},
are also maximal packing broadcast on $C_n$, with $n\ge 3$, $n=8q+r$ and $0\le r\le 7$, 
which gives $p_b(C_n)\leq \sigma(f_r)$, that is, $p_b(C_n)$ is not greater than the value claimed
by the theorem.

We now prove the opposite inequality.
Observe first that $010$, $2000$, $20000$ and $001001$
are optimal solutions for $C_3$, $C_4$, $C_5$ and $C_6$, respectively,
and that their cost is exactly the value claimed by the theorem.
We can thus assume $n\ge 7$.
Let $f$ be a   $p_b$-broadcast on $C_n$, $n\ge 7$.

We first claim that we necessarily have $|V^+_f|\geq 2$.
Indeed, suppose to the contrary that $|V^+_f|=1$ and
that $V^+_f=\{x_0\}$, without loss of generality.
Since $f$ is maximal, we necessarily have 
$\sigma(f)=f(x_0)=\left\lfloor\frac{n-1}{2}\right\rfloor$,
and thus $p_b(C_n)=\left\lfloor\frac{n-1}{2}\right\rfloor$,
which contradicts the inequality we previously established when $n\ge 7$.

We thus have $|V^+_f|\geq 2$.
Let $V^+_f=\{x_{i_0},\dots,x_{i_{t-1}}\}$, $t\ge 2$.
%
Since $f$ is maximal, we necessarily have 
$$f(x_{i_j}) + f(x_{i_{j+1}}) +1 \leq  d_{C_n}(x_{i_j},x_{i_{j+1}}) \leq f(x_{i_j}) + f(x_{i_{j+1}}) + 3$$
for every $j$, $0\le j\le t-1$.
Moreover, if 
$d_{C_n}(x_{i_j},x_{i_{j+1}}) \ge f(x_{i_j}) + f(x_{i_{j+1}}) + 2$, then we necessarily have
$d_{C_n}(x_{i_{j+1}},x_{i_{j+2}}) = f(x_{i_{j+1}}) + f(x_{i_{j+2}}) + 1$
(we may have $i_j = i_{j+2}$), since otherwise we could
increase the value of $f(x_{i_{j+1}})$ by~1.

We consider the two following cases.
\begin{enumerate}
\item If all vertices in $C_n$ are $f$-dominated, 
that is, $d_{C_n}(x_{i_{j}},x_{i_{j+1}}) = f(x_{i_{j}}) + f(x_{i_{j+1}}) + 1$ for every $j$, $0\le j\le t-1$,
then, to avoid confusion, we let $P_n=y_0y_1\dots y_{n-1}$
with $y_0=x_{{i_0}+f(x_{i_0})+1}$.
Let now $g$ be the function
defined by $g(y_j)=f(x_{j+i_0+f(x_{i_0})+1})$ for every $j$, $0\le j\le n-1$.
Clearly, both $y_0$ and $y_{n-1}$ are $g$-dominated.
Since $f$ was a maximal packing broadcast on $C_n$, we get that $g$ is also
a maximal packing broadcast on $P_n$, which gives $p_b(P_n)\leq p_b(C_n)$. 

\item
If $C_n$ contains at least one vertex which is not $f$-dominated, then
we can assume, without loss
of generality, that $x_{i_1+f(x_{i_1})+1}$ is not $f$-dominated,
which implies $d_{C_n}(x_{i_{0}},x_{i_{1}}) = f(x_{i_{0}}) + f(x_{i_{1}})+1$
and $d_{C_n}(x_{i_{2}},x_{i_{3}}) = f(x_{i_{2}}) + f(x_{i_{3}})+1$
(we may have $x_{i_0}=x_{i_2}$ and $x_{i_1}=x_{i_3}$).
Again, to avoid confusion, we let $P_n=y_0y_1\dots y_{n-1}$
with $y_0=x_{{i_1}+f(x_{i_1})+2}$.
Let now $g$ be the function
defined by $g(y_j)=f(x_{j+i_1+f(x_{i_1})+2})$ for every $j$, $0\le j\le n-1$.
Since $d_{C_n}(x_{i_{0}},x_{i_{1}}) = f(x_{i_{0}}) + f(x_{i_{1}})+1$
and $d_{C_n}(x_{i_{2}},x_{i_{3}}) = f(x_{i_{2}}) + f(x_{i_{3}})+1$, 
the values of both the leftmost and the rightmost $g$-broadcast vertex in $P_n$ cannot be increased.
Since $f$ was a maximal packing broadcast on $C_n$, we thus get that $g$ is also
a maximal packing broadcast on $P_n$, which gives $p_b(P_n)\leq p_b(C_n)$.
\end{enumerate}


We thus have $p_b(P_n)\leq p_b(C_n)$ in all cases and the result follows. 
\end{proof}

\end{document}